\documentclass[11pt, letterpaper]{article}
\usepackage[utf8]{inputenc} %codification of the document
\usepackage[toc, page]{appendix}
\usepackage{graphicx}

\usepackage[a4paper,width=140mm,top=40mm,bottom=40mm,bindingoffset=6mm]{geometry}
\usepackage{fancyhdr}
\usepackage[english]{babel}

\usepackage{indentfirst}
\usepackage[style = alphabetic, backend=bibtex, maxbibnames=10]{biblatex}
\usepackage{subfig}
\usepackage{comment}
\usepackage{enumerate}
\usepackage{amsmath}

\usepackage{amsfonts}

\newcommand{\overbar}[1]{\mkern 1.5mu\overline{\mkern-1.5mu#1\mkern-1.5mu}\mkern 1.5mu}

\usepackage{commath}
\usepackage{systeme}
\usepackage{esint}
\usepackage{amssymb}
\usepackage{amsthm}
\usepackage[hidelinks]{hyperref}
\usepackage{xcolor}
\hypersetup{
colorlinks,
linkcolor={blue!50!black},
urlcolor={blue!80!black},
citecolor={blue!30!black}}

\makeatletter
\def\blfootnote{\xdef\@thefnmark{}\@footnotetext}
\makeatother

\numberwithin{equation}{section}
\newtheorem{theorem}{Theorem}[section]
\newtheorem{corollary}[theorem]{Corollary}

\newtheorem{proposition}[theorem]{Proposition}

\newtheorem{remark}[theorem]{Remark}

\theoremstyle{theorem}
\newtheorem{assumption}[theorem]{Assumption}
\newtheorem{definition}[theorem]{Definition}

\newcommand{\pp}{\mathbb{P}}

\newcommand{\rr}{\mathbb{R}}

\newcommand{\E}{\mathbb{E}}

\newcommand{\eq}{\begin{equation}}
\newcommand{\en}{\end{equation}}

\begin{document}

\thispagestyle{plain}
\begin{center}
\Large\textbf{Construction of Forward Performance Processes in Stochastic Factor Models and an Extension of Widder's Theorem}

\vspace{0.4cm}
\large
Levon Avanesyan \quad Mykhaylo Shkolnikov\footnote{M. Shkolnikov was partially supported by the NSF grant DMS-1506290.} \quad Ronnie Sircar

\vspace{0.4cm}
\textit{Princeton University}

\vspace{0.9cm}
\textbf{Abstract}
\end{center}
We consider the problem of optimal portfolio selection under forward investment performance criteria in an incomplete market. Given multiple traded assets, the prices of which depend on multiple observable stochastic factors, we construct a large class of forward performance processes with power-utility initial data, as well as the corresponding optimal portfolios. This is done by solving the associated non-linear parabolic partial differential equations (PDEs) posed in the ``wrong'' time direction, for stock-factor correlation matrices with eigenvalue equality (EVE) structure, which we introduce here. Along the way we establish on domains an explicit form of the generalized Widder's theorem of Nadtochiy and Tehranchi \cite[Theorem 3.12]{nadtochiy2015optimal} and rely hereby on the Laplace inversion in time of the solutions to suitable linear parabolic PDEs posed in the ``right'' time direction.

%The dynamics of prices of the traded assets depend on $k$ observable stochastic factors. We construct a class of Forward Performance Process with power-utility initial data by reducing the problem to solving an ill-posed parabolic equation. We establish a link between solutions to well-posed and ill-posed parabolic equations through a Laplace transform, thereby providing an explicit of way to do such construction.

%%%%%%%%%%%%%%%%%%%%%%
\section{Introduction}
%%%%%%%%%%%%%%%%%%%%%%

In this paper we study the optimal portfolio selection problem under forward investment criteria in incomplete markets, specifically stochastic factor models. Our setup is that of a continuous-time market model with multiple stocks whose growth rates and volatilities are functions of multiple observable stochastic factors following jointly a diffusion process. The incompleteness arises hereby from the imperfect correlation between the Brownian motions driving the stock prices and the factors. The factors themselves can model various market inputs, including stochastic interest rates, stochastic volatility and major macroeconomic indicators, such as inflation, GDP growth or the unemployment rate.

\medskip

The optimal portfolio problem in continuous time was originally considered by Merton in his pioneering work \cite{merton1969lifetime}, \cite{merton1971optimum}, and is commonly referred to as the Merton problem. In this framework an investor looks to maximize her expected terminal utility from wealth acquired in the investment process within a geometric Brownian motion market model. Good compilations of classical results can be found in the books \cite{duffie2010dynamic}, \cite{karatzas1998methods}. As fundamental as this setup is, it has two important drawbacks. First, the investor must decide on her terminal utility function before entering the market, and thereby cannot adapt it to
changes in market conditions. Second, before settling on an investment strategy, the investor must firmly set her time horizon. That is, the portfolio derived in this framework is optimal only for one specific utility function over one time horizon.

\medskip

External factors such as the economic cycle, natural disasters, and the political climate can lead to dynamic changes in one's level of risk aversion. This would change the terminal utility function, thereby affecting the optimal portfolio allocation. Even if the terminal utility function stays the same, the investor might decide to exit the market at an earlier or a later time than originally planned. For two investment horizons $0<T_1<T_2$ there is no natural relation between the two respective optimal portfolios. Thus, if the investor initially decided to stay in the market until time $T_1$, but later on decided to continue the investment activities until time $T_2$, she would have to either incur significant transaction costs to rebalance her portfolio, or continue investing at a suboptimal level of expected utility from terminal wealth. In both cases she would regret her past decisions, thereby making the classical approach
terminal time inconsistent. We call performance criteria \textit{terminal time consistent} if the optimal dynamic portfolio on the time interval $[0, T_2]$ restricted to the interval $[0,T_1]$ yields the optimal dynamic portfolio on the time interval $[0, T_1]$. Finding such criteria is essential in solving portfolio optimization problems with an uncertain investment horizon. For this purpose \textit{forward investment performance criteria} were introduced and developed in \cite{musiela2006investments} and \cite{musiela2007investment}, as well as in \cite{HENDERSON20071621}.

\medskip
%Come back to the comments of this paragraph after reading Thalea's, Sergiy's and El Karoui's papers again.

Instead of looking to optimize the expectation of a deterministic utility function at a single terminal point in time, this approach looks to maximize the expectation of a stochastic utility function at every single point in time. \textit{Forward performance processes} (FPPs) capture the time evolutions of such stochastic utility functions. They are increasing and strictly concave in the wealth argument, intrinsically incorporate the randomness stemming from the market, and most importantly yield terminal time-consistent investment strategies. Other than completely specifying the market and the factors that affect it, the only piece of information a portfolio manager needs is the investor's initial utility function. The portfolio manager can infer the shape of this function (or, equivalently, the level of risk aversion) by observing the return targets and the error bounds around them set by the investor.

\medskip

A comprehensive description of all FPPs remains a challenging open problem. Much work towards this goal has been carried out throughout the last ten years, see \cite{berrier2009characterization}, \cite{nicole2013exact}, \cite{nicole2013stochastic}, \cite{HENDERSON20071621}, \cite{musiela2010stochastic}, and \cite{zitkovic2009} for some important results. In \cite{musiela2010stochastic}, Musiela and Zariphopoulou proposed a construction of FPPs by means of solutions to a stochastic partial differential equation (SPDE). The SPDE can be thought of as the
forward stochastic analogue of the Hamilton-Jacobi-Bellman (HJB) equation that arises in the optimization of the expected utility from terminal wealth. Every classical solution of this SPDE which is increasing and strictly concave in the wealth argument is a local FPP, but no existence theory for such SPDEs is available, and additional conditions (to be checked on a case-by-case basis) are needed to ensure that the local FPP is a true FPP. The key novelty and difficulty in dealing with this SPDE is the introduction of the forward volatility process. It reflects the investor's uncertainty about her preferences in the future and is subject to her choice. To find all the FPPs characterized by the SPDE, one would have to find all forward volatility processes, along with initial utility functions, for which the SPDE has a classical solution. The case of zero forward volatility yields time-monotone FPPs, and was extensively discussed in \cite{musiela2010portfolio2} and \cite{musiela2010portfolio}. In \cite{nicole2013exact} and \cite{nicole2013stochastic}, El Karoui and M'rad find a functional representation of the forward volatility for which, given an initial utility function and a wealth process satisfying certain regularity conditions, the SPDE has a classical solution. Moreover, if the solution is a true FPP, it renders the chosen wealth process optimal. This is an important result, as it helps to infer investors' performance criteria from the portfolios they pick in a given market. Here, we are concerned with the complementary problem of constructing an FPP \textit{and} an associated optimal portfolio for an investor entering a new market equipped with her initial utility function.

\medskip

We consider factor-driven market models and FPPs into which the randomness enters only through the underlying stochastic factors. Assuming such a form, with a compatible forward volatility process, the SPDE mentioned above reduces to an HJB equation set in the ``wrong'' time direction. We will call its classical solutions \textit{factor-form} local FPPs if they are increasing and strictly concave in the wealth argument. In a complete market one can use the Fenchel-Legendre transform to linearize the HJB equation, and arrive at a linear second-order parabolic PDE set in the ``wrong'' time direction (see \cite{nadtochiy2015optimal}). In an incomplete market no such linearizing transformation is available in general. To the best of our knowledge, the only exception is the special case of power utility in a one-factor market model, where a linearization is possible through a distortion transformation, as discovered in \cite{zariphopoulou2001solution} for the Merton problem, and used for the construction of FPPs in \cite{nadtochiy2015optimal}, \cite{nadtochiy2014class}, and \cite{shkolnikov2015asymptotic}. We show that for a multiple factor market model with a special stock-factor correlation matrix structure (see Assumption \ref{assump.Corr} below) the distortion transformation still simplifies the HJB equation to a linear second-order parabolic equation set in the ``wrong'' time direction.

\medskip
Motivated by such a simplification in one-factor market models, Nadtochiy and Tehranchi \cite[Theorem 3.12]{nadtochiy2015optimal}  exhibited a characterization of all positive solutions to such linear parabolic equations. Their theorem constitutes a generalization of the celebrated Widder's theorem (see \cite{widder1963Appel}), which describes all positive solutions of the heat equation set in the ``wrong'' time direction. The generalized Widder's theorem reveals that positive solutions of a linear second-order parabolic equation set in the ``wrong'' time direction must be linear combinations of exponentially scaled positive eigenfunctions for the corresponding elliptic operator according to a positive finite Borel measure. Moreover, each solution is uniquely identified with a pairing of the eigenfunctions and the measure.

\medskip
In our first main theorem (Theorem \ref{Thm.FPP.Construct}) we give a new version of \cite[Theorem 3.12]{nadtochiy2015optimal} on domains
in the multiple stocks multiple factor setup with an initial utility function of power type to describe a new class of FPPs. Note that generalized Widder's theorems do not provide a way to construct the pairings of the eigenfunctions and the measure. Our second set of results (see Theorem \ref{Thm.Suff.Cond} and Remark \ref{rmk:inv.Lapl}) addresses this issue: in Theorem \ref{Thm.Suff.Cond} we give the Laplace transform of the measure in terms of the solution to a linear parabolic equation set in the ``right'' time direction, and we provide a method (see Remark
\ref{rmk:inv.Lapl}) of finding the only possible corresponding eigenfunctions as well. Thus, we indeed obtain a large explicit class of FPPs.

\medskip

The rest of the paper is structured as follows. In Section \ref{sec:MainResults} we state our main results, postponing their proofs to later sections.
In Section \ref{sec:FPPs} we introduce relevant facts about FPPs and subsequently prove Theorem \ref{Thm.FPP.Construct}. In Section
\ref{sec:Elliptic.Eigenfunctions} we show Theorem \ref{Thm.Suff.Cond}, summarize some results from the theory of linear elliptic operators, and use them to establish Propositions \ref{Prop.Dim1.Critical}, \ref{Prop.Dim1.Multi} and \ref{Prop.DimInfty.Multi}. In Section \ref{sec:BackwardCase} we discuss the Merton problem within the framework of our market model. Lastly, in Section \ref{sec: corr matr} we discuss the meaning of the main assumption in Theorem \ref{Thm.FPP.Construct} (Assumption \ref{assump.Corr}).

%%%%%%%%%%%%%%%%%%%%%%%%%%%%%
\section{Main results}\label{sec:MainResults}
%%%%%%%%%%%%%%%%%%%%%%%%%%%%%

%%%%%%%%%%%%%%%%%%%%%%%%%%%%%
\subsection{Model}
%%%%%%%%%%%%%%%%%%%%%%%%%%%%%

Consider an investor with initial capital $X_0=x>0$ aiming to invest in a market with $n\ge 1$ stocks, the prices of which follow a process $S$, and a riskless bank account with zero interest rate. The stock prices depend on an observable $k$-dimensional stochastic factor process $Y$ taking values in $D \subseteq \mathbb{R}^k$, and are driven by a $d_W$-dimensional standard Brownian motion $W$. The factor process $Y$ is itself driven by a $d_B$-dimensional standard Brownian motion $B$, whose correlation with $W$ is given by a matrix $\mathrm{corr}(W,B) = (\rho_{ij})_{i,j=1}^{d_W, d_B}$ with singular values in $[0,1]$. Without loss of generality we assume that $d_W\ge n$ (see \cite[Remark 0.2.6]{karatzaslectures}). The investor's filtration $({\mathcal F}_t)_{t\ge0}$ is generated by a pair $(S,Y)$ of processes satisfying
\begin{eqnarray}
&& \frac{\mathrm{d}S^i_t}{S_t^i}=\mu_i(Y_t)\,\mathrm{d}t+\sum_{j=1}^{d_W} \sigma_{ji}(Y_t)\,\mathrm{d}W^j_t,\quad i=1,\,2,\,\ldots,\,n, \label{eq:Dynamic.Stock} \\
&& \dif Y_t = \alpha(Y_t) \dif t + \kappa(Y_t)^T\dif B_t,
\label{eq:Dynamic.Factor} \\
&& B_t = \rho^T W_t + A^T W^{\perp}_t, \label{eq:BM}
\end{eqnarray}
where the superscript $T$ denotes transposition and $W^\perp$ is a $d_{W^\perp}$-dimensional standard Brownian motion independent of $W$. We write $\mu$ for $(\mu_1,\mu_2,\ldots,\mu_n)^T$ and $\sigma$ for $(\sigma_{ij})_{i,j=1}^{d_W,n}$ throughout.

\medskip

For the convenience of the reader we summarize the dimensions of all the quantities we have introduced thus far:
\begin{eqnarray*}
&& \mu(\cdot) - n\times 1,\;\quad \sigma(\cdot) - d_W\times n, \;\quad W_t - d_W\times 1, \\
&& \alpha(\cdot) - k\times 1,\;\quad \kappa(\cdot) - d_B\times k,\;\;\quad B_t - d_B\times 1, \\
&& \rho - d_W\times d_B,\quad A - d_{W^{\perp}}\times d_B,\quad W^{\perp}_t - d_{W^{\perp}}\times 1.
\end{eqnarray*}
Note that there is no loss of generality in using the representation \eqref{eq:BM} for the standard Brownian motion $B$, since we can let $A$ be the square root of the positive semidefinite matrix $I_{d_B}-\rho^T\rho$ (recall that the singular values of $\rho$ belong to $[0,1]$), and $d_{W^\perp}=d_B$.

\begin{assumption}\label{assump.Var}
The functions $\mu:\,D\to\rr^n$, $\sigma:\,D\to\rr^{d_W\times n}$ are continuous, the stochastic differential equation (SDE) \eqref{eq:Dynamic.Factor} possesses a unique weak solution, and the columns of $\rho$ belong to the range of left-multiplication by $\sigma(y)$ for all $y\in D$.
\end{assumption}

\begin{remark}\label{rmk:sigma.sigma.inverse}
Under Assumption \ref{assump.Var} it holds $\sigma(y)\sigma(y)^{-1}\rho=\rho$ for all $y\in D$, where $\sigma(y)^{-1}$ is the Moore-Penrose pseudoinverse of $\sigma(y)$. Indeed, $\sigma(y)\sigma(y)^{-1}\sigma(y)=\sigma(y)$, so that the columns of $\sigma(y)$ (and consequently the vectors in their span, that is, the range of the left-multiplication by $\sigma(y)$) are invariant under the left-multiplication by $\sigma(y)\sigma(y)^{-1}$. This is true, in particular, if $\sigma(y)$ has rank $n$ for all $y\in D$.
\end{remark}

%Come up with an example that will discuss the case when the Brownian Motions moving the stocks and
%the factors are independent.

%\begin{remark}
%The matrix $\sigma(\cdot)$ having a full rank implies that $\sigma \sigma^{-1} (\cdot) =  I_{d_W\times d_W} (\cdot)$,
%\end{remark}

Our investor dynamically allocates her wealth in the market using a
self-financing trading strategy that at any time $t\ge0$ yields a portfolio allocation $\pi_t = (\pi^1_t,\dots,\pi^n_t)$ among the $n$ stocks with the associated wealth process
\begin{equation} \label{eq:Dynamic.Wealth}
\frac{\dif X^{\pi}_t}{X^{\pi}_t} = (\sigma(Y_t)\pi_t)^T \lambda(Y_t)\dif t + (\sigma(Y_t)\pi_t)^T \dif W_t,\quad  X^\pi_0  = x,
\end{equation}
where $\lambda(Y_t)=(\sigma(Y_t)^T)^{-1}\mu(Y_t)$ is the Sharpe ratio. Apart from the self-financeability, we impose additional conditions on the trading strategies to ensure that their wealth processes $X^\pi$ are well-defined by \eqref{eq:Dynamic.Wealth}.

%Finding such a strategy implies solving a certain optimization problem. To do that, we have to
%first define a set of strategies over which we will be maximizing our objective function.

\begin{definition}\label{defn:AdmissibilitySet}
An $\mathcal{F}_t$-progressively measurable self-financing trading strategy is called admissible if its portfolio allocation $\pi$ among the $n$ stocks fulfills
\begin{equation}\label{eq:AdmissibilitySet}
\forall\,t\ge0:\quad \int_0^t \big|\pi_s^T\sigma(Y_s)^T\lambda(Y_s)\big|\,\mathrm{d}s<\infty\quad\text{and}\quad
\int_0^t \big|\sigma(Y_s)\pi_s\big|^2\,\mathrm{d}s<\infty
\end{equation}
with probability one. In this case, we write $\pi\in{\mathcal A}$.
\end{definition}

Next, we define (local) forward performance processes, which capture how the utility functions of an investor evolve over time as she continues to invest in the financial market above. Part of the definition is an optimality criterion for portfolio allocations $\pi\in{\mathcal A}$ that reflects the dynamic programming principle time-consistent optimal portfolio allocations $\pi^*\in\mathcal{A}$ must satisfy.

%Now we need to characterize an objective function for which the optimization problem yielding the
%above-mentioned desired properties is feasible. One could attain this by choosing it from a class
%of stochastic functions called Forward Performance Processes.
\begin{definition}\label{defn:FPP}
An $\mathcal{F}_t$-progressively measurable $U_{\cdot}(\cdot):\,[0,\infty)\times(0,\infty)\to\rr$ is referred to as a (local) forward performance process (FPP) if
\begin{enumerate}[(i)]
\item with probability one, all functions $x\mapsto U_t(x)$, $t\ge0$ are strictly concave and increasing,
\item for each $\pi\in\mathcal{A}$, the process $U_t(X_t^{\pi})$, $t\ge0$ is an $(\mathcal{F}_t)_{t\ge0}$ (local) supermartingale,
\item there exists an optimal $\pi^*\in\mathcal{A}$ for which $U_t(X_t^{\pi^*})$, $t\ge0$ is an $(\mathcal{F}_t)_{t\ge0}$ (local) martingale.
\end{enumerate}
\end{definition}

%%%%%%%%%%%%%%%%%%%%%%%%%%
\subsection{Separable power factor form FPPs in EVE models}
%%%%%%%%%%%%%%%%%%%%%%%%%%

We consider (local) FPPs of factor-form into which the randomness enters only through the stochastic factor process, that is,
\begin{equation}\label{eq:FPP.Factor}
U_t(x)=V(t,x,Y_t),\quad t\geq 0
\end{equation}
for a deterministic function $V:\,[0,\infty)\times(0,\infty)\times D\to\rr$. To be able to construct functions $V$ such that the corresponding $U_\cdot(\cdot)$ is a (local) FPP in the generality of the setup  \eqref{eq:Dynamic.Stock}, \eqref{eq:Dynamic.Factor} we focus on the situation when the initial utility function is of product form and a power function in the wealth variable:
\begin{equation}\label{eq:FPP.Factor.Separable.Power}
U_0(x) = V(0,x,Y_0) = \gamma^\gamma\frac{x^{1-\gamma}}{1-\gamma}\,h(Y_0)\quad \textrm{for some }\;\gamma\in(0,\infty)\backslash\{1\}.
\end{equation}

\begin{remark}\label{rmk:sep.factor.form}
The crucial simplification arising from the structure in \eqref{eq:FPP.Factor.Separable.Power} lies in its propagation to positive times. In fact, we will construct (local) FPPs of the form
\begin{equation}\label{def:sep.power}
U_t(x)=V(t,x,Y_t)=\gamma^\gamma\frac{x^{1-\gamma}}{1-\gamma}\,g(t,Y_t),
\end{equation}
where $g$ is continuously differentiable in $t$ (its first argument) and twice continuously differentiable in $y$ (the second argument). We propose to call them \emph{separable power factor form (local) FPPs}.
\end{remark}

We are able to characterize \textit{all} separable power factor form local FPPs under the next assumption on the correlation matrix $\rho=\mathrm{corr}(W,B)$.

\begin{assumption}\label{assump.Corr}
For some $p\in[0,1]$,
\begin{equation} \label{eq:Assump.Corr}
\rho^T\rho = p\,I_{d_B}.
\end{equation}
\end{assumption}

\begin{remark}\label{rmk:EVE_disc}
For any orthonormal $d_B\times d_B$ matrix $O$, we may replace $\kappa(\cdot)$ by $O\kappa(\cdot)$ and $B$ by $\widetilde{B}=OB$ in \eqref{eq:Dynamic.Factor} without changing the dynamics of the pair $(S,Y)$. Since $\widetilde{B}$ is a $d_B$-dimensional standard Brownian motion and $\mathrm{corr}(W,\widetilde{B})=O^T\rho^T\rho O$ is diagonal for an appropriate choice of $O$, we could have assumed without loss of generality from the very beginning that $\rho^T\rho$ is diagonal. Thus, the only true restriction imposed by Assumption \ref{assump.Corr} lies in the equality of the eigenvalues of $\rho^T\rho$. We refer to market models that satisfy the condition \eqref{eq:Assump.Corr} as eigenvalue equality (EVE) models. Note that for EVE models, since $\rho$ is a $d_W\times d_B$-matrix, at least one of the following two has to hold true:
\begin{enumerate}[(i)]
\item $d_W \geq d_B$,
\item $p =0$.
\end{enumerate}
Finally, we remark that when $d_B = 1$, $\rho^T\rho$ is a scalar, so that Assumption \ref{assump.Corr} holds automatically. Section \ref{sec: corr matr} is devoted to a further discussion of EVE models.
\end{remark}

\subsection{Characterizing the FPPs}\label{subsec:CharFPP}

In order to describe our construction of separable power factor form FPPs, we need to introduce some quantities related to linear elliptic operators of the second order. Consider on $C^2(D)$ such an operator
\begin{equation}\label{eq:EllipticOperator}
\mathcal{L} = \frac{1}{2}\sum_{i,j=1}^{k}a_{ij}(y)\frac{\partial^2}{\partial y_i\partial y_j} +
 \sum_{i=1}^{k}b_i(y)\frac{\partial}{\partial y_i} + P(y)
\end{equation}
under the following assumption.

\begin{assumption}\label{Assump.Holder}
The operator $\mathcal{L}$ is locally uniformly elliptic with locally $\eta$-H\"older continuous and globally bounded coefficients. That is, with $a(\cdot)=(a_{ij}(\cdot))_{i,j=1}^k$ and $b(\cdot)=(b_1(\cdot),b_2(\cdot),\ldots,b_k(\cdot))^T$, there exists an $\eta\in(0,1)$ such that, for any bounded subdomain $D'$ of $D$ satisfying $\overbar{D'}\subset D$,
\begin{enumerate}[(i)]
\item $\inf_{y \in  \overbar{D'},\, |v| = 1}v^T a(y) v >0$,
\item $\|a\|_{\eta, D'},\,\|b\|_{\eta, D'},\,\|P\|_{\eta, D'} < \infty$, where $\|f\|_{\eta,D'}= \sup_{x,y\in D',x\neq y} \frac{|f(x)-f(y)|}{|x-y|^{\eta}}$,
\end{enumerate}
and
\begin{enumerate}
\item[(iii)] $\sup_{y\in D} |a(y)|,\,\sup_{y\in D} |b(y)|,\,\sup_{y\in D} |P(y)| < \infty$.
\end{enumerate}
\end{assumption}

\begin{remark}
Whenever $D=\rr^k$ and conditions (i)-(iii) in Assumption \ref{Assump.Holder} hold with $\rr^k$ instead of $\overbar{D'}$ and $D'$, they are also fulfilled in their original form. Moreover, in this case, the SDE \eqref{eq:Dynamic.Factor} has a unique weak solution (see \cite[Chapter 5, Remarks 4.17 and 4.30]{karatzas1991brownian}).
\end{remark}

%\begin{remark}
%  Consider the generator of our factor process $\mathcal{L}_y$. If it satisfies the Assumption \ref{Assump.Holder}, then the SDE  has a unique weak solution (see Corollary 2.9 in \cite{karatzas1991brownian}).
%\end{remark}
  %The main results of this paper are about constructing power utility factor-form Forward
  %Performance Processes in a general factor-market setting summarized in
  %\eqref{eq:Dynamic.Stock}-\eqref{eq:Assump.Corr}.

We define the H\"older space $C^{2,\eta}(D)\subset C^2(D)$ as the subspace consisting of functions whose second-order partial derivatives are locally $\eta$-H\"older continuous (in the same sense as in condition (ii) of Assumption \ref{Assump.Holder}). Next, we introduce the sets of positive eigenfunctions for the operator $\mathcal{L}$, which correspond to eigenvalues $\zeta\in\rr$, and are normalized at some fixed $y_0 \in D$:
\begin{equation}
C_{\mathcal{L} - \zeta}(D) = \big\{\psi \in C^{2,\eta}(D):\, \psi(\cdot)>0,\, \psi(y_0) = 1,\, (\mathcal{L} - \zeta)\psi=0 \big\}.
\end{equation}
Moreover, we let $\mathbb{S}_{\mathcal{L}}(D)$ be the spectrum of $\mathcal{L}$ associated with positive eigenfunctions:
\begin{equation}
\mathbb{S}_{\mathcal{L}}(D)= \big\{\zeta \in \mathbb{R}:\: C_{\mathcal{L} - \zeta}(D)\neq\emptyset\big\}.
\end{equation}
Finally, we call a functional $\Psi:\,\mathbb{S}_{\mathcal{L}}(D)\times D\to(0,\infty)$ such that $\Psi(\zeta,\cdot) \in C_{\mathcal{L} - \zeta}(D)$ for all $\zeta\in\mathbb{S}_{\mathcal{L}}(D)$, a \textit{selection of positive eigenfunctions}, and recall the definition of Bochner integrability in this setting.

\begin{definition}
Given a positive finite Borel measure $\nu$ on $\mathbb{S}_{\mathcal{L}}(D)$, we refer to a selection of positive eigenfunctions $\Psi:\,\mathbb{S}_{\mathcal{L}}(D)\times D \to(0,\infty)$ as $\nu$-Bochner integrable if, for all compact $K\subset D$, $\int_{\mathbb{S}_{\mathcal{L}}(D)} \|\Psi(\zeta, \cdot)\|_K\,\nu(\dif \zeta) < \infty$, where $\|f\|_K=\sup_{y\in K} |f(y)|$.
\end{definition}

We are now ready to state our first main result.

\begin{theorem}\label{Thm.FPP.Construct}
Suppose the market model \eqref{eq:Dynamic.Stock}, \eqref{eq:Dynamic.Factor}, the correlation matrix $\rho$, and the linear elliptic operator of the second order $\mathcal{L}$ in \eqref{eq:EllipticOperator} with the coefficients
\begin{equation}\label{L_coeff}
a(\cdot) = \kappa(\cdot)^T\kappa(\cdot),\quad
b(\cdot) = \alpha(\cdot) + \Gamma\kappa(\cdot)^T\rho^T\lambda(\cdot),\quad
P(\cdot) = \frac{\Gamma}{2q}\lambda(\cdot)^T\lambda(\cdot),
\end{equation}
where $\Gamma =\frac{1-\gamma}{\gamma}$ and $q = \frac{1}{1+ \Gamma p}$, satisfy the Assumptions \ref{assump.Var}, \ref{assump.Corr}, and \ref{Assump.Holder}, respectively. Then:

\begin{enumerate}[(i)]
\item For any positive finite Borel measure $\nu$ on $\mathbb{S}_{\mathcal{L}}(D)$ and a $\nu$-Bochner integrable selection of positive eigenfunctions $\Psi:\,\mathbb{S}_{\mathcal{L}}(D)\times D \to (0,\infty)$ the unique separable power factor form local FPP $U_\cdot(\cdot)$ with the initial condition
\begin{equation}\label{eq:Thm.FPP.Construct.IC}
U_0(x) = \gamma^\gamma \frac{x^{1-\gamma}}{1-\gamma}\,\bigg( \int_{\mathbb{S}_{\mathcal{L}}(D)}\Psi(\zeta,Y_0)\,\nu (\dif \zeta)\bigg)^q
\end{equation}
is given by
\begin{equation}\label{eq:Thm.FPP.Construct.Sol}
U_t(x) = \gamma^{\gamma} \frac{x^{1-\gamma}}{1-\gamma}\,\bigg(\int_{\mathbb{S}_{\mathcal{L}}(D)} e^{-t\zeta}
  \Psi(\zeta,Y_t)\,\nu(\dif \zeta)\bigg)^q.
\end{equation}
Moreover, any $\pi^*$ that solves
\begin{equation}\label{eq:Thm.FPP.Construct.OptPort}
  \sigma(Y_t) \pi_t^{*} =
  \frac{1}{\gamma}\,\Bigg(\!\lambda(Y_t) + q \rho \kappa(Y_t)\,
  \frac{\int_{\mathbb{S}_{\mathcal{L}}(D)} e^{-t\zeta}\,(\nabla_y \Psi)(\zeta, Y_t)\,\nu(\dif
  \zeta)}{\int_{\mathbb{S}_{\mathcal{L}}(D)} e^{-t\zeta}\,\Psi(\zeta, Y_t)\,\nu(\dif \zeta)} \Bigg)
  % This follows by plugging U(t,x) into the expression for the optimal portfolio from Nadtochiy
  %and Tehranchi: -\frac{\lambda(Y_t) V_x(t,x,Y_t) + \sigma(Y_t) \nabla_y V_x(t,x, Y_t) }{x
  %V_{xx}(t,x,Y_t)}\\
\end{equation}
is an associated optimal portfolio.
\item Given a function $h:\,D\to(0,\infty)$, there exists a local FPP of separable power factor form with the initial condition
\begin{equation}\label{eq:U0}
U_0(x)=\gamma^\gamma \frac{x^{1-\gamma}}{1-\gamma}\,h(Y_0)^q
\end{equation}
if and only if there exists a positive finite Borel measure $\nu$ on $\mathbb{S}_{\mathcal{L}}(D)$ and a $\nu$-Bochner integrable selection of positive eigenfunctions $\Psi:\,\mathbb{S}_{\mathcal{L}}(D)\times D \to (0,\infty)$ such that
\begin{align}\label{eq:Thm.FPP.Construct.Condition}
h(y) = \int_{\mathbb{S}_{\mathcal{L}}(D)} \Psi(\zeta, y)\,\nu (\dif \zeta).
\end{align}
In this case, the local FPP of separable power factor form and the corresponding optimal portfolios are given by
\eqref{eq:Thm.FPP.Construct.Sol} and \eqref{eq:Thm.FPP.Construct.OptPort}, respectively.
\end{enumerate}
\end{theorem}

%\begin{remark}
%Note that imposing Assumption \ref{Assump.Holder} on the operator $\mathcal{L}$ directly enforces the same assumption on the generator of our factors $\mathcal{L}_y$, thereby yielding weak uniqueness for the factor SDE \ref{eq:Dynamic.Factor}.
%  \end{remark}

\begin{remark}
We note that the equation \eqref{eq:Thm.FPP.Construct.OptPort} for optimal portfolios $\pi^*$ does not involve the initial wealth $x$. This is a consequence of the local FPP being of separable power factor form. In the setting of the Merton problem, the same statement is true (and well-known) for terminal utility functions of power form.
\end{remark}

\begin{remark}
A solution to the optimal portfolio equation \eqref{eq:Thm.FPP.Construct.OptPort} can be obtained as follows. Since $\sigma(\cdot)^{-1}=(\sigma(\cdot)^T\sigma(\cdot))^{-1}\sigma(\cdot)^T$, one can write $\lambda(\cdot)=(\sigma(\cdot)^T)^{-1}\mu(\cdot)$ as $\sigma(\cdot)(\sigma(\cdot)^T\sigma(\cdot))^{-1}\mu(\cdot)$. In addition, by Assumption \ref{assump.Var} and the Borel selection result of \cite[Theorem 6.9.6]{Bog}, one can find a measurable $\varsigma:\,D\to\rr^{n\times d_B}$ satisfying $\sigma(\cdot)\varsigma(\cdot)=\rho$, which renders
\begin{equation}\label{optimal_portfolio}
\pi^*_t=\frac{1}{\gamma}\,\Bigg(\!(\sigma(Y_t)^T\sigma(Y_t))^{-1}\mu(Y_t) + q \varsigma(Y_t)\kappa(Y_t)\,
  \frac{\int_{\mathbb{S}_{\mathcal{L}}(D)} e^{-t\zeta}\,(\nabla_y \Psi)(\zeta, Y_t)\,\nu(\dif
  \zeta)}{\int_{\mathbb{S}_{\mathcal{L}}(D)} e^{-t\zeta}\,\Psi(\zeta, Y_t)\,\nu(\dif \zeta)} \Bigg)
\end{equation}
a solution of \eqref{eq:Thm.FPP.Construct.OptPort}.
\end{remark}

Part (ii) of Theorem \ref{Thm.FPP.Construct} shows that, once a portfolio manager has an estimate for an investor's level of risk-aversion $\gamma$ and the functional dependence (encoded by $h$) of her current utility function on the value $Y_0$ of the factor process, he can extrapolate the future values of her utility function according to \eqref{eq:Thm.FPP.Construct.Sol} and acquire a portfolio fulfilling \eqref{eq:Thm.FPP.Construct.OptPort} (e.g. the portfolio in \eqref{optimal_portfolio}) on her behalf, provided $h$ is of the form \eqref{eq:Thm.FPP.Construct.Condition}. It is therefore crucial to understand which functions $h$ admit the representation \eqref{eq:Thm.FPP.Construct.Condition} and to be able to determine the pairings $(\Psi,\nu)$ for such.

%%%%%%%%%%%%%%%%%%%%%%%%%%%%%%%
\subsection{Finding selections of positive eigenfunctions $\Psi$ and measures $\nu$}
%%%%%%%%%%%%%%%%%%%%%%%%%%%%%%%

The next set of results addresses the problem of solving the equation \eqref{eq:Thm.FPP.Construct.Condition} for the pairing $(\Psi,\nu)$, when it exists. The equation \eqref{eq:Thm.FPP.Construct.Condition} stems from a further generalization of the generalized Widder's theorem of Nadtochiy and Tehranchi \cite[Theorem 3.12]{nadtochiy2015optimal} (see Theorem \ref{Prop:GeneralizedWidder} below) and, thus, our results can be viewed as yielding explicit versions of such theorems. The following theorem is also of independent interest, as it relates the pairing $(\Psi,\nu)$ arising in the positive solution of a linear second-order parabolic PDE posed in the ``wrong'' time direction to the solution of the same PDE posed in the ``right'' time direction.

%The previous way of writing up the necessary condition for h to be of desired form

%\begin{proposition}\label{Prop.Suff.Cond}
%   Let $\mathcal{L}$ satisfy Assumption \ref{Assump.Holder} and $Z_t$ be the diffusion generated by it. For some positive Borel measure $\nu$ a function $h\in C^2(D)$ takes on a form $h(y) = \int_{\mathbb{S}_{\mathcal{L}}}\Psi(\zeta, y) \nu (\dif \zeta)$ only if
%   $t \to E_0\bigg[h(Z_t) e^{\int_{0}^{t}P(Z_s)\dif s}\bigg]$ is a Laplace transform of $\nu$.
%\end{proposition}

\begin{theorem}\label{Thm.Suff.Cond}
Let $\mathcal{L}$ satisfy Assumption \ref{Assump.Holder} and let $h\in C^{2,\eta}(D)$ be a positive function such that
\begin{equation}\label{asmp.FK}
(t,y)\mapsto\E\big[h(Z_t)\,\mathbf{1}_{\{\tau>t\}}\,\big|\,Z_0=y\big]
\end{equation}
is locally bounded on $[0,\epsilon]\times D$ for the weak solution $Z$ of the SDE associated with ${\mathcal L}_0:={\mathcal L}-P(y)$ and $\varepsilon>0$, where $\tau$ is the first exit time of $Z$ from $D$. Then, there exists a classical solution to
\begin{equation}\label{what is u}
\partial_t u + \mathcal{L} u=0\;\;\text{on}\;\;[-\varepsilon,0]\times D,\quad u(0,\cdot)=h.
\end{equation}
Moreover, for a positive finite Borel measure $\nu$ on $\mathbb{S}_{\mathcal{L}}(D)$ and a $\nu$-Bochner integrable selection of positive eigenfunctions $\Psi:\,\mathbb{S}_{\mathcal{L}}(D)\times D \to (0,\infty)$ the function $h$ can be expressed as $\int_{\mathbb{S}_{\mathcal{L}}(D)}\Psi(\zeta,\cdot)\,\nu(\dif \zeta)$ if and only if, for every $y\in D$, the function $u(\cdot,y)$ on $(-\varepsilon,0]$ is the Laplace transform of the measure $\Psi(\zeta,y)\,\nu(\mathrm{d}\zeta)$, that is,
\begin{equation}
u(t,y) = \int_{\mathbb{S}_{\mathcal{L}}(D)} e^{-\zeta t}\,\Psi(\zeta,y)\,\nu (\dif \zeta),\quad t\in(-\varepsilon,0].
\end{equation}
In this case, it holds, in particular,
\begin{equation}\label{eq:Lapl_conn}
u(t,y_0) = \int_{\mathbb{S}_{\mathcal{L}}(D)} e^{-\zeta t}\,\nu (\dif \zeta),\quad t\in(-\varepsilon,0].
\end{equation}
\end{theorem}

\begin{remark}\label{rmk:inv.Lapl}
Theorem \ref{Thm.Suff.Cond} reveals that, whenever a pairing $(\Psi,\nu)$ exists, it can be inferred by finding the measure $\nu$ through a one-dimensional Laplace inversion of $u(\cdot,y_0)$ (recall that the values of the Laplace transform on a non-trivial interval determine the underlying positive finite Borel measure, see \cite[Section 30]{billingsley2012probability}) and then the functions $\Psi(\cdot,y)$, $y\in D\backslash\{y_0\}$ from $u(\cdot,y)$, $y\in D\backslash\{y_0\}$ through additional one-dimensional Laplace inversions.
\end{remark}

As a by-product we obtain the following uniqueness result for linear second-order parabolic PDEs posed in the ``wrong'' time direction by combining the generalized Widder's theorem on domains (Theorem \ref{Prop:GeneralizedWidder} below) with Theorem \ref{Thm.Suff.Cond} and the uniqueness of the Laplace transform (see \cite[Section 30]{billingsley2012probability}).

\begin{corollary} \label{Prop:GeneralizedWidderUniqueness}
For any operator $\mathcal{L}$ satisfying Assumption \ref{Assump.Holder} and positive $h\in C^{2,\eta}(D)$ such that the function in  \eqref{asmp.FK} is locally bounded on a non-trivial cylinder $[0,\varepsilon]\times D$, there is at most one positive solution $\tilde{u}$ of the problem
\begin{equation}\label{eq:ill_posed}
\partial_t \tilde{u} + \mathcal{L} \tilde{u}=0\;\;\text{on}\;\;[0,\infty)\times D,\quad\tilde{u}(0,\cdot)=h.
\end{equation}
\end{corollary}

\begin{remark}
We stress that Corollary \ref{Prop:GeneralizedWidderUniqueness} is not an immediate consequence of the generalized Widder's theorem on domains (Theorem \ref{Prop:GeneralizedWidder}) by itself. The latter does ensure that every pairing $(\Psi,\nu)$ corresponds to exactly one positive solution $\tilde{u}$ of \eqref{eq:ill_posed}. However, it is not clear a priori whether the representation $h=\int_{\mathbb{S}_{\mathcal{L}}(D)} \Psi(\zeta,\cdot)\,\nu(\dif \zeta)$ is unique for all functions $h$ with the property \eqref{asmp.FK}. Theorem \ref{Thm.Suff.Cond} and the uniqueness of the Laplace transform (see \cite[Section 30]{billingsley2012probability}) show that this representation is, indeed, unique.
\end{remark}

For arbitrary operators relatively little is known about the sets of positive eigenfunctions $C_{\mathcal{L}-\zeta}(D)$. Nevertheless, in certain situations additional information on the sets $C_{\mathcal{L}-\zeta}(D)$ is available and can be exploited to find the selection of positive eigenfunctions $\Psi$ for a given function $h$ by a \textit{finite} number of Laplace inversions.
%In the best-case scenario, $|C_{\mathcal{L}-\zeta}(D)|=1$ for all $\zeta\in\mathrm{supp}(\nu)$, so that the measure $\nu$ determines the selection $\Psi$ and no additional Laplace inversions are required. The next results summarize conditions guaranteeing $|C_{\mathcal{L}-\zeta}(D)|=1$.

\begin{proposition}\label{Prop.Dim1.Critical}
Let $\mathcal{L}$ satisfy Assumption \ref{Assump.Holder}, then
\begin{equation}\label{crit_eig}
\zeta_c(D):=\inf\big\{\zeta\in\rr:\,
\zeta \in \mathbb{S}_{\mathcal{L}}(D)\big\}\in \mathbb{S}_{\mathcal{L}}(D).
\end{equation}
If, in addition, the potential $P$ is constant and $\mathcal{L}_0:=\mathcal{L}-P$ is such that the corresponding solution of the generalized martingale problem on $D$ (see \cite[Section 1.13]{pinsky1995positive}) is recurrent, then $\zeta_c(D)=-P$ and $|C_{\mathcal{L} - \zeta_c(D)}(D)|=1$.
\end{proposition}

\begin{remark}
The quantity $\zeta_c(D)$ of \eqref{crit_eig} is commonly referred to as the \emph{critical eigenvalue} of the operator $\mathcal{L}$ on $D$.
\end{remark}

The structure of the eigenspaces $C_{\mathcal{L - \zeta}}(D)$ can differ widely depending on the choice of the dimension $k$, the restrictions on the operator $\mathcal{L}$, and the domain $D$. The case $k=1$ corresponds to having a single factor and leads to eigenspaces of dimension at most $2$.

\begin{proposition}\label{Prop.Dim2.One}
Suppose $\mathcal{L}$ satisfies Assumption \ref{Assump.Holder} on a domain $D\subset\rr$. Then, the number of extreme points of the convex set $C_{\mathcal{L}-\zeta}(D)$ is $2$ for all $\zeta>\zeta_c(D)$ and belongs to $\{1,2\}$ for $\zeta=\zeta_c(D)$.
\end{proposition}

\begin{remark}
Proposition \ref{Prop.Dim2.One} reveals that, in the setting of Theorem \ref{Thm.Suff.Cond} with $k=1$, one can determine the pairing $(\Psi,\nu)$ via a three-step procedure: first, one recovers $\nu$ by a one-dimensional Laplace inversion of $u(\cdot,y_0)$; second, one finds $\Psi(\zeta,y_1)\,\nu(\mathrm{d}\zeta)$ by a one-dimensional Laplace inversion of $u(\cdot,y_1)$ for an arbitrary $y_1\in D\backslash\{y_0\}$; third, for all $\zeta\ge\zeta_c(D)$, one solves the second-order linear \emph{ordinary} differential equation for $\Psi(\zeta,\cdot)$ with the obtained boundary conditions at $y_0$ and $y_1$ to end up with the selection $\Psi$.
\end{remark}
When $k \geq 2$, the variability in the dimensionality of the eigenspaces is illustrated by the following two scenarios, in which the eigenspaces have dimensions $1$ and $\infty$, respectively.

\begin{definition}\label{assump.Fuchs}
A potential $P(\cdot)$ on $\mathbb{R}^k$ is called \emph{principally radially symmetric} if
\begin{equation}
P = P_0 + P_1,
\end{equation}
where the functions $P_0$ and $P_1$ are locally integrable to power $d$ for some $d>k/2$, with $P_0$ being radially symmetric ($P_0(y)=\tilde{P}_0(|y|)$ for some $\tilde{P}_0$), and $P_1$ vanishing outside of a compact set.
\end{definition}

\begin{proposition}\label{Prop.Dim1.Multi}
Consider a positive $\phi\in C^{2,\eta}(\rr^k)$ with bounded $\frac{\nabla\phi}{\phi}$ and $\frac{\Delta\phi}{\phi}$, as well as an operator $\tilde{\mathcal{L}}:=\Delta + P(y)$ on $\rr^k$ with a locally $\eta$-H\"older  continuous bounded principally symmetric potential $P(\cdot)$. Then, $\mathcal{L}:=\frac{1}{\phi}\tilde{\mathcal{L}}\phi$ has the property $|C_{\mathcal{L}-\zeta}(\rr^k)|=1$ for any $\zeta>\zeta_c(\rr^k)$ such that
\begin{equation}\label{eq:Prop.Dim1.Multi.IntegralCond}
  \int_{1}^{\infty}t^{k-3}g_0(t)^2\bigg(\int_{t}^{\infty} s^{1-k} g_0(s)^{-2} \dif s \bigg) \dif t = \infty,
\end{equation}
where $g_0$ is the unique solution of
\begin{equation}\label{eq:Prop.Dim1.Multi.ODE}
g_0''(r)+\frac{k-1}{r}g_0'(r)+\big(\zeta-\tilde{P}_0(r)\big)g_0(r)=0\;\;\textrm{on}\;\;(0,\infty),\quad g_0(r) = 1+o(r)\;\;\textrm{as}\;\;r\downarrow 0.
\end{equation}

%Let $\mathcal{L}(f) = \frac{1}{\phi}\tilde{\mathcal{L}}(\phi f)$ for all $f \in C^{2,\alpha}(\rr^k)$, where $\phi > 0$ is in $C^{2,\alpha}(\rr^k)$, and $\tilde{\mathcal{L}} := \Delta + P(y)$ is an operator satisfying Assumption \ref{Assump.Holder} with a principally radially symmetric potential $P$ on $\rr^k$ such that. Then, $|C_{\mathcal{L}-\zeta}(\rr^k)|=1$ for all $\zeta>\zeta_c(D)$.
\end{proposition}

In the situation of Proposition \ref{Prop.Dim1.Multi}, we must pick $\Psi(\zeta,\cdot)$ as the unique element of $C_{\mathcal{L}-\zeta}(\rr^k)$. On the other hand, in the case of a multidimensional factor process on a bounded domain $D$ with a Lipschitz boundary, the eigenspaces are infinite-dimensional.

% \begin{assumption}[B]\label{assump.Bounded}
%We call the operator $\mathcal{L}$ bounded on $D\subset\mathbb{R}^k$ if there exist constants
%$c_1,\, c_2 >0$ such that for all $y \in \mathbb{R}^k$
%\begin{enumerate}[(i)]
%  \item  $c_1 \leq \inf_{y \in  \overbar{D'},\, \|v\| = 1}v^T a(y) v,\quad \sup_{y \in
%      \overbar{D'},\, \|v\| = 1}v^T a(y) v \leq c_2$,
%  \item $|b_i(y)|,\, |P(y)| \leq c_2,\quad \forall i \in \overbar{[1,k]},$
%\end{enumerate}
%where $r[a(\cdot)],\,R[a(\cdot)]$ are the smallest and largest eigenvalues of matrix $a(\cdot)$
%respectively.
%\end{assumption}

\begin{proposition}\label{Prop.DimInfty.Multi}
Let $D\subset\mathbb{R}^k$, $k\ge 2$ be a bounded domain with a Lipschitz boundary and $\mathcal{L}$ satisfy (i)-(iii) in Assumption \ref{Assump.Holder} with $D$ instead of $\overbar{D'}$ and $D'$. Then, the convex set $C_{\mathcal{L}-\zeta}(D)$ has infinitely many extreme points for all $\zeta>\zeta_c(D)$.
\end{proposition}

Thus, one cannot assert that the number of extreme points of $C_{\mathcal{L}-\zeta}(D)$ is finite in the generality of Assumption \ref{Assump.Holder}. Therefore, the procedure of Remark \ref{rmk:inv.Lapl} cannot always be reduced to a finite number of Laplace inversions. In such cases, we propose to determine the selection $\Psi$ on a finite number of grid points $y\in D$.

%%%%%%%%%%%%%%%%%%%%%%%%%%%%%%%%%
\section{Proof of Theorem \ref{Thm.FPP.Construct} and a new Widder's theorem}\label{sec:FPPs}
%%%%%%%%%%%%%%%%%%%%%%%%%%%%%%%%%

The goal of this section is to prove Theorem \ref{Thm.FPP.Construct}. Recall that we are interested in separable power factor form local FPPs defined in Remark \ref{rmk:sep.factor.form}. We start by focusing on the function $V$ and give a sufficient condition for $V(t,x,Y_t)$ to be a local FPP.

\begin{proposition}\label{Thm.FPP.Sufficient}
Under Assumption \ref{assump.Var} let $V:\,[0,\infty)\times(0,\infty)\times D\to\rr$ be continuously differentiable in $t$ (its first argument) and twice continuously differentiable in $x$ and $y$ (the second and third arguments). Suppose further that $V$ is strictly concave and increasing in $x$ and a classical solution of the HJB equation
\begin{equation}\label{eq:Thm.FPP.Sufficient.PDE.HJB}
\partial_t V + \mathcal{L}_y V - \frac{1}{2}\,
\frac{|\lambda\,\partial_x V + \rho\kappa\,\partial_x \nabla_y V|^2}{\partial_{xx} V} = 0 \quad\text{on}\quad [0,\infty)\times(0,\infty)\times D,
\end{equation}
where $\mathcal{L}_y$ is the generator of the factor process $Y$. Then, $V(t,x, Y_t)$ is a local FPP. Moreover, the corresponding optimal portfolio allocations $\pi^*$ among the $n$ stocks are of a feedback form and characterized by
\begin{equation}\label{eq:Thm.FPP.Sufficient.OptPortfolio}
\sigma(Y_t)\pi_t^* = -\frac{\lambda(Y_t)\,\partial_x V(t,X_t^{\pi^*},Y_t) + \rho\kappa(Y_t)\,\partial_x \nabla_y V(t,X_t^{\pi^*},Y_t)}{X_t^{\pi^*}\partial_{xx}V(t,X_t^{\pi^*},Y_t)}.
\end{equation}
%The statement with U's.

%Let $U(t,x)$ be twice continuously differentiable in $x$ and Ito process $\in \mathcal{F}_t$ such that the mapping $x\to U(t,x)$ is increasing and concave. Let, also $U(t,x)$ be a solution to the stochastic partial differential equation (SPDE):
%\begin{align}\begin{split}\label{eq:FPP.SPDE}
%  \dif U(t,\,x) = \frac{1}{2} \frac{\|U_x(t,\,x)\lambda(t) + \sigma(t)\sigma^{-1}(t)a_x^{W}(t,\,x)\|^2}{U_{xx}(t,\,x)} \dif &t \\
%   + v (t,\,x)^T\dif \mathcal{W}(t),\quad U(0,x) = u_0(&x),
%   \end{split}
%   \end{align}
%where $v(t,x) = (v^W, v^{W^{\perp}})\in \mathcal{F}_t$ is progressively measurable. Then, $U(t,x)$ is an FPP subject to some technical conditions.
\end{proposition}

\begin{proof}
For the former statement, one only needs to repeat the derivation of \cite[equation (1.6)]{shkolnikov2015asymptotic} mutatis mutandis and to use $\sigma(\cdot)\sigma(\cdot)^{-1}\rho=\rho$ (see Remark \ref{rmk:sigma.sigma.inverse}). For the latter statement, we apply It\^o's formula to $V(t,X^\pi_t,Y_t)$ and substitute $\frac{1}{2}\,
\frac{|\lambda\,\partial_x V + \rho\kappa\,\partial_x \nabla_y V|^2}{\partial_{xx} V}$ for $\partial_t V + \mathcal{L}_y V$ to conclude that the drift coefficient of $V(t,X^\pi_t,Y_t)$ is the negative of
\begin{equation}\label{eq:drift.coeff}
\frac{1}{2}\bigg|\frac{\lambda(Y_t)\partial_x V(t,X_t^\pi,Y_t)\!+\!\rho\kappa(Y_t)\partial_x\nabla_y V(t,X_t^\pi,Y_t)}{(-\partial_{xx}V(t,X_t^\pi,Y_t))^{1/2}}-X^\pi_t\sigma(Y_t)\pi_t(-\partial_{xx}V(t,X_t^\pi,Y_t))^{1/2}\bigg|^2\!.
\end{equation}
The process $V(t,X^\pi_t,Y_t)$ is a local martingale if and only if the expression in \eqref{eq:drift.coeff} vanishes, which happens if and only if
\eqref{eq:Thm.FPP.Sufficient.OptPortfolio} holds.
\end{proof}

\begin{remark}
The process $V(t,x,Y_t)$ of Proposition \ref{Thm.FPP.Sufficient} is a true FPP if $\,V(t,X^\pi_t,Y_t)$ is a true supermartingale for every $\pi\in{\mathcal A}$ and a true martingale for every optimal portfolio allocation $\pi^*$ of \eqref{eq:Thm.FPP.Sufficient.OptPortfolio}. In view of Fatou's lemma, the supermartingale property is fulfilled if $\inf_{s\in[0,t]} V(s,X^\pi_s,Y_s)$ is integrable for all $t\ge0$ and $\pi\in{\mathcal A}$. The martingale property is valid if the diffusion coefficients $\partial_x V(t,X^{\pi^*}_t,Y_t)X^{\pi^*}_t(\sigma(Y_t)\pi^*_t)^T$, $\nabla_y V(t,X^{\pi^*}_t,Y_t)\kappa(Y_t)^T$ of $V(t,X^{\pi^*}_t,Y_t)$ are $\mathrm{d}t\times\mathrm{d}\pp$-square integrable on each $[0,t]\times\Omega$.
\end{remark}

The HJB equation \eqref{eq:Thm.FPP.Sufficient.PDE.HJB} is a fully non-linear PDE and one does not expect to find explicit formulas for its solutions in general. However, for initial conditions of separable power type and under the Assumption \ref{assump.Corr}, the HJB equation \eqref{eq:Thm.FPP.Sufficient.PDE.HJB} can be linearized.

\begin{proposition}\label{lem:FPP.Exists.Parab.PDE}
Let Assumption \ref{assump.Corr} be satisfied, $\Gamma=\frac{1-\gamma}{\gamma}$, and $q=\frac{1}{1+\Gamma p}$. Then, the HJB equation \eqref{eq:Thm.FPP.Sufficient.PDE.HJB} with an initial condition $V(0,x,y)=\gamma^\gamma \frac{x^{1-\gamma}}{1-\gamma}h(y)^q$, where $h\!>\!0$, has a classical solution in separable power form, $V(t,x,y)\!=\!\gamma^\gamma \frac{x^{1-\gamma}}{1-\gamma}g(t,y)$, with $g>0$ if and only if there exists a positive solution to the linear PDE problem
\begin{equation}\label{problem.for.Phi}
\partial_t u+\mathcal{L}u=0\;\;\text{on}\;\;[0,\infty)\times D,\quad u(0,\cdot)=h
\end{equation}
posed in the ``wrong'' time direction. Hereby, $\mathcal{L}$ is the linear elliptic operator of the second order with the coefficients of \eqref{L_coeff}. In that case, the two solutions are related through
\begin{equation}\label{eq:lem:FPP.Exists.Parab.PDE.Explicit}
V(t,x,y)=\gamma^\gamma\frac{x^{1-\gamma}}{1-\gamma}u(t,y)^q.
\end{equation}
\end{proposition}

\begin{proof}
Since we are looking for solutions of the HJB equation \eqref{eq:Thm.FPP.Sufficient.PDE.HJB} in separable power form, we plug in the ansatz $V(t,x,y)=\gamma^{\gamma}\frac{x^{1-\gamma}}{1-\gamma} g(t,y)$ to arrive at
\begin{equation} \label{eq:MultiDPDEg}
\partial_t g + \mathcal{L}_y g + \frac{\Gamma}{2}\lambda^T\lambda g + \Gamma\lambda^T\rho\kappa\nabla_y g + \Gamma\frac{(\nabla_y g)\kappa^T \rho^T \rho\kappa\nabla_y g}{2g}=0,\quad g(0,\cdot)=h^q.
\end{equation}
Next, we employ the distortion transformation $g(t,y) = u(t,y)^q$ and get the PDE
\begin{equation}\label{Phi.a.priori}
\begin{split}
&\;qu^{q-1}\partial_tu
+\frac{1}{2}\sum_{i,j=1}^{k}(\kappa^T\kappa)_{ij} \big(qu^{q-1}\partial_{y_i,y_j}u+q(q-1)u^{q-2}(\partial_{y_i}u)(\partial_{y_j}u)\big) \\
& +\frac{\Gamma}{2}q^2u^{q-2}(\nabla_yu)^T\kappa^T\rho^T\rho\kappa \nabla_yu
+q\big(\alpha+\Gamma\kappa^T\rho^T\lambda\big)^Tu^{q-1}\nabla_yu
+\frac{\Gamma}{2}\lambda^T\lambda u^q=0,
\end{split}
\end{equation}
equipped with the initial condition $u(0,\cdot)=h$. Moreover, the assumed positivity of $g$ translates to $u>0$, so that we can divide both sides of \eqref{Phi.a.priori} by $u^{q-1}$. In addition, we insert the identity $\rho^T\rho = p I_{d_B}$ of Assumption \ref{assump.Corr} to end up with
\begin{equation}\label{Phi.interm}
\begin{split}
&\;\partial_t u+ \frac{1}{2}\sum_{i,j=1}^{k}(\kappa^T\kappa)_{ij}\partial_{y_i y_j}u
+ \big(\alpha + \Gamma\kappa^T \rho^T \lambda\big)^T\nabla_y u
+ \frac{\Gamma}{2q}\lambda^T\lambda u \\
& + \frac{1}{2u}(q+\Gamma pq-1)(\nabla_y u)^T\kappa^T\kappa\nabla_y u=0.
\end{split}
\end{equation}
The crucial observation is now that the non-linear term in the PDE \eqref{Phi.interm} drops out thanks to $q=\frac{1}{1+\Gamma p}$. Hence, $u$ is a positive solution of \eqref{problem.for.Phi}. The converse follows by carrying out the transformations we have used in the reverse order.
\end{proof}

Proposition \ref{lem:FPP.Exists.Parab.PDE} reduces the task of finding solutions of the HJB equation \eqref{eq:Thm.FPP.Sufficient.PDE.HJB} in separable power form to solving the linear PDE problem \eqref{problem.for.Phi} set in the ``wrong'' time direction. The latter has been studied in \cite{widder1963Appel} with $\mathcal{L}$ being the Laplace operator on $\rr^k$ and in \cite{nadtochiy2015optimal} for more general linear second-order elliptic operators on $\rr^k$. We establish subsequently a further generalization of \cite[Theorem 3.12]{nadtochiy2015optimal} that allows for linear second-order elliptic operators on arbitrary domains $D\subset\rr^k$.

\begin{theorem}\label{Prop:GeneralizedWidder}
Under Assumption \ref{Assump.Holder} a function $u:\,\{(0,y_0)\}\cup((0,\infty)\times D)\to(0,\infty)$ is a classical solution of $\partial_t u+\mathcal{L}u=0$ with $u(0,y_0)=1$ if and only if it admits the representation
\begin{equation}\label{eq:Thm.GeneralizedWidder.}
u(t,y) = \int_{\mathbb{S}_\mathcal{L}(D)} e^{-t\zeta}\,\Psi(\zeta,y)\,\nu(\dif \zeta),
\end{equation}
where $\nu$ is a Borel probability measure on $\mathbb{S}_{\mathcal{L}}(D)$ and $\Psi:\,\mathbb{S}_{\mathcal{L}}(D)\times D\to(0,\infty)$ is a $\nu$-Bochner integrable selection of positive eigenfunctions. In this case, the pairing $(\Psi,\nu)$ is uniquely determined by the function $u$.
\end{theorem}

\begin{proof}
We can adapt the proof of \cite[Theorem 3.12]{nadtochiy2015optimal} to the situation at hand. Consider any subdomain $D'\subset D$ satisfying $y_0\in D'$ and $\overbar{D'}\subset D$. We endow the space of continuous functions on $\{(0,y_0)\}\cup((0,\infty)\times D')$ with the topology of uniform convergence on the compact subsets of the sets
\begin{equation}
\big\{(t,y)\in[0,\infty)\times D':\,t\ge c|y-y_0|^2\big\},\quad c>0.
\end{equation}
Next, we repeat the proofs of \cite[Theorem 3.6, Lemmas 3.7, 3.9, 3.10, and Theorem 3.11]{nadtochiy2015optimal} and the necessity part of the proof of \cite[Theorem 3.12]{nadtochiy2015optimal}, just replacing their $\rr^n$ by our $D'$ and the Harnack's inequality employed therein by the one in \cite[Chapter VII, Corollary 7.42]{Lieberman}, to deduce that every function $u$ as in the statement of the theorem can be expressed as
\begin{equation}\label{D'repr}
\int_\rr e^{-t\zeta}\,\Psi_{D'}(\zeta,y)\,\nu_{D'}(\dif \zeta)\;\;\text{for}\;\; (t,y)\in\{(0,y_0)\}\cup((0,\infty)\times D')\to(0,\infty),
\end{equation}
with a Borel probability measure $\nu_{D'}$ on $\rr$ and $\Psi_{D'}(\zeta,\cdot)\in C_{\mathcal{L}-\zeta}(D')$, $\zeta\in\mathrm{supp}(\nu_{D'})$. This conclusion for a sequence of the described subdomains $D'$ increasing to $D$ and the uniqueness of the Laplace transform (see \cite[Section 30]{billingsley2012probability}) imply that \eqref{D'repr} applies with the same $\nu$ and $\Psi(\zeta,\cdot)\in C_{\mathcal{L}-\zeta}(D')$, $\zeta\in\mathrm{supp}(\nu)$ for all $D'$ in the sequence, so that \eqref{eq:Thm.GeneralizedWidder.} and the uniqueness of the pairing $(\Psi,\nu)$ readily follow. Conversely, proceeding as in the sufficiency part of the proof of \cite[Theorem 3.12]{nadtochiy2015optimal} we find that, for every subdomain $D'$ as above, the right-hand side of \eqref{eq:Thm.GeneralizedWidder.} is a classical solution of $\partial_t u+\mathcal{L}u=0$ on $\{(0,y_0)\}\cup((0,\infty)\times D')$ with $u(0,y_0)=1$. Picking a sequence of subdomains $D'$ increasing to $D$ as before we obtain the sufficiency part of Theorem \ref{Prop:GeneralizedWidder}.
\end{proof}

We now have all the ingredients needed to prove Theorem \ref{Thm.FPP.Construct}.

\begin{proof}[Proof of Theorem \ref{Thm.FPP.Construct}] \textit{(i).} Take a pairing $(\Psi,\nu)$ as specified in point (i) of the theorem. By Proposition \ref{Thm.FPP.Sufficient} it is enough to provide a classical solution $V(t,x,y)=\gamma^\gamma\frac{x^{1-\gamma}}{1-\gamma}g(t,y)$ of the HJB equation \eqref{eq:Thm.FPP.Sufficient.PDE.HJB} with the properties as
in that proposition and satisfying the initial condition
\begin{equation}
V(0,x,y) = \gamma^\gamma \frac{x^{1-\gamma}}{1-\gamma}\bigg(\int_{\mathbb{S}_{\mathcal{L}}(D)} \Psi(\zeta,y)\,\nu(\dif\zeta)\bigg)^q,\quad (x,y)\in(0,\infty)\times D.
\end{equation}
In view of Proposition \ref{lem:FPP.Exists.Parab.PDE}, such a function $V$ can be constructed by solving
\begin{equation}\label{Phi.problem}
\partial_t u+\mathcal{L}u=0\;\;\text{on}\;\;[0,\infty)\times D\;\;\text{with}\;\;u(0,\cdot)=\int_{\mathbb{S}_{\mathcal{L}}(D)} \Psi(\zeta,\cdot)\,\nu(\dif\zeta)
\end{equation}
and inserting the solution $u$ into the right-hand side of \eqref{eq:lem:FPP.Exists.Parab.PDE.Explicit}. By Theorem \ref{Prop:GeneralizedWidder}, the solution $u$ of \eqref{Phi.problem} is given by the right-hand side of \eqref{eq:Thm.GeneralizedWidder.}.

\medskip

Conversely, for a separable power factor form local FPP $\gamma^\gamma\frac{x^{1-\gamma}}{1-\gamma}g(t,Y_t)$ and a portfolio allocation $\pi\in\mathcal{A}$, we apply It\^o's formula to $\gamma^\gamma\frac{(X^\pi_t)^{1-\gamma}}{1-\gamma}g(t,Y_t)$ and infer from the conditions (ii) and (iii) in Definition \ref{defn:FPP} that the resulting drift coefficient must be non-positive for all $\pi\in{\mathcal A}$ and equal to $0$ for any maximizer $\pi^*\in{\mathcal A}$. Equating the maximum of the drift coefficient over all $\pi\in{\mathcal A}$ to $0$ we end up with the PDE in \eqref{eq:MultiDPDEg} for $g$. Moreover, the proof of Proposition \ref{lem:FPP.Exists.Parab.PDE} reveals that the function $u$ associated with $g$ via $g(t,y)=u(t,y)^q$ solves the problem \eqref{problem.for.Phi} with $h(\cdot)=\int_{\mathbb{S}_{\mathcal{L}}(D)}\Psi(\zeta,\cdot)\,\nu (\dif \zeta)$. At this point, the identity \eqref{eq:Thm.FPP.Construct.Sol} follows from Theorem \ref{Prop:GeneralizedWidder}. Finally, the characterization \eqref{eq:Thm.FPP.Construct.OptPort} of the optimal portfolios is a direct consequence of \eqref{eq:Thm.FPP.Sufficient.OptPortfolio} and \eqref{eq:Thm.FPP.Construct.Sol}.

\medskip

\noindent\textit{(ii).} Arguing as in the second half of the proof of part (i) we deduce that, for any separable power factor form local FPP $\gamma^\gamma\frac{x^{1-\gamma}}{1-\gamma}g(t,Y_t)$ with the initial condition of \eqref{eq:U0}, the function $g$ is a classical solution of the problem \eqref{eq:MultiDPDEg}. The substitution $g(t,y)=u(t,y)^q$ and Theorem \ref{Prop:GeneralizedWidder} show the necessity and sufficiency of the representation \eqref{eq:Thm.FPP.Construct.Condition}. We conclude as in the second half of the proof of part (i).
\end{proof}

%%%%%%%%%%%%%%%%%%%%%%%%%%%%%%%%
\section{Proof of Theorem \ref{Thm.Suff.Cond} and further ramifications}\label{sec:Elliptic.Eigenfunctions}
%%%%%%%%%%%%%%%%%%%%%%%%%%%%%%%%

%%%%%%%%%%%%%%%%%%%%%%%%%%%%%%%%
\subsection{Proof of Theorem \ref{Thm.Suff.Cond}}

We start our analysis of the pairing $(\Psi,\nu)$ by establishing Theorem \ref{Thm.Suff.Cond}.

\begin{proof}[Proof of Theorem \ref{Thm.Suff.Cond}]
Let $D'\subset D$ be a bounded subdomain with a $C^3$ boundary $\partial D'\subset D$ and $\psi:\,D'\to[0,1]$ be a thrice continuously differentiable function with compact support in $D'$. Then, by \cite[Chapter IV, Theorem 5.2]{Lady} the problem
\begin{equation}
\partial_t u_{D'} + \mathcal{L} u_{D'}=0\;\;\text{on}\;\;[-\varepsilon,0]\times D',\quad u_{D'}|_{[-\varepsilon,0]\times \partial D'}=0,\quad u_{D'}(0,\cdot)=h\psi
\end{equation}
(posed in the ``right'' time direction) has a unique classical solution with $\eta$-H\"older continuous $\partial_t u_{D'}$, $\partial_{y_iy_j} u_{D'}$ in the $y$ variable, $\frac{\eta}{2}$-H\"older continuous $\partial_t u_{D'}$, $\partial_{y_iy_j} u_{D'}$ in the $t$ variable, and $\frac{1+\eta}{2}$-H\"older continuous $\partial_{y_i} u_{D'}$ in the $t$ variable. In particular, $u_{D'}$ obeys the Feynman-Kac formula
\begin{equation}
u_{D'}(-t,y)=\E\Big[e^{\int_0^t P(Z_s)\,\mathrm{d}s}\,(h\psi)(Z_t)\,\mathbf{1}_{\{\tau_{D'}>t\}}\,\Big|\,Z_0=y\Big],\quad (t,y)\in[0,\varepsilon]\times D',
\end{equation}
where $\tau_{D'}$ is the first exit time of $Z$ from $D'$.

\medskip

Using the described construction for a sequence of subdomains $D'$ and functions $\psi$ increasing to $D$ and $\mathbf{1}_D$, respectively, we arrive at the monotone limit
\begin{equation}
u(-t,y)=\E\Big[e^{\int_0^t P(Z_s)\,\mathrm{d}s}\,h(Z_t)\,\mathbf{1}_{\{\tau_D>t\}}\,\Big|\,Z_0=y\Big],\quad (t,y)\in[0,\varepsilon]\times D
\end{equation}
of $u_{D'}$, which is locally bounded on $[0,\varepsilon]\times D$ by assumption. Thanks to this and the local regularity estimate in \cite[Chapter IV, Theorem 10.1]{Lady} we can extract a subsequence of $u_{D'}$ converging uniformly together with $\partial_t u_{D'}$, $\partial_{y_i} u_{D'}$, and $\partial_{y_iy_j} u_{D'}$ on every fixed set $[-\varepsilon,0]\times D'$. Thus, $u$ is a classical solution of the problem \eqref{what is u}.

\medskip

Now, assume $h=\int_{\mathbb{S}_{\mathcal{L}}(D)} \Psi(\zeta,\cdot)\,\nu(\dif\zeta)$. In view of \cite[Chapter 4, Theorem 3.2 and Exercise 4.16]{pinsky1995positive} (see also Section \ref{subsec:PrelimPositiveEigen} for more details), the elements of $\mathbb{S}_{\mathcal{L}}(D)$ are bounded below, so that the function $\tilde{u}(t,y)=\int_{\mathbb{S}_{\mathcal{L}}(D)} e^{-\zeta t}\,\Psi(\zeta,y)\,\nu(\dif\zeta)$ is well-defined on $[0,\infty)\times D$. By Theorem \ref{Prop:GeneralizedWidder}, the function $\tilde{u}$ is a classical solution of
\begin{equation}\label{sys:Thm.Suff.Cond.Proof.Parab.WrongTime}
\partial_t\tilde{u}+\mathcal{L}\tilde{u}=0\;\;\text{on}\;\;\{(0,y_0)\}\cup ((0,\infty)\times D).
\end{equation}
Moreover, the function
\begin{equation}\label{sys:Thm.Suff.Cond.Proof.Define.v}
v(t,y) =
\begin{cases}
u(t,y)\;\;\text{for}\;\;(t,y)\in[-\varepsilon,0]\times D, \\
\tilde{u}(t,y)\;\;\text{for}\;\;(t,y)\in(0,\infty)\times D
\end{cases}
\end{equation}
is a classical solution of the PDE $\partial_t v + \mathcal{L} v = 0$ on $[-\varepsilon,\infty)\times D$. Indeed, on the set $([-\varepsilon,0)\cup(0,\infty))\times D$ this PDE holds by construction, whereas
\begin{equation}
\partial_t \tilde{u}(0,y)=\lim_{t\downarrow0}\,\partial_t\tilde{u}(t,y)=- \lim_{t\downarrow0}\,\mathcal{L}\tilde{u}(t,y)=-\mathcal{L}\tilde{u}(0,y),\quad y\in D
\end{equation}
by the interior Schauder estimate of \cite[Theorem 6.2]{nadtochiy2015optimal}.

\medskip

The Harnack's inequality in \cite[Chapter VII, Corollary 7.42]{Lieberman} enables us to apply Theorem \ref{Prop:GeneralizedWidder} to the function
\begin{equation}
\tilde{v}:\,\{(0,y_0)\}\cup ((0,\infty)\times D)\to(0,\infty),\quad (t,y)\mapsto\frac{v(t-\varepsilon,y)}{v(-\varepsilon,y_0)}
\end{equation}
and find a Borel probability measure $\tilde{\nu}$ on $\mathbb{S}_{\mathcal{L}}(D)$ and a $\tilde{\nu}$-Bochner integrable selection of positive eigenfunctions $\tilde{\Psi}:\,\mathbb{S}_{\mathcal{L}}(D)\times D\to(0,\infty)$ such that
\begin{equation}
\tilde{v}(t,y)=\int_{\mathbb{S}_{\mathcal{L}}(D)} e^{-t\zeta}\,\tilde{\Psi}(\zeta,y)\,\tilde{\nu}(\mathrm{d}\zeta).
\end{equation}
In particular, for $(t,y)\in(0,\infty)\times D$,
\begin{equation}\label{eq:Thm.Suff.Cond.Proof.NewIntegral.v}
\int_{\mathbb{S}_{\mathcal{L}}(D)} e^{-(t+\varepsilon)\zeta}\,\tilde{\Psi}(\zeta,y)\,\tilde{\nu}(\mathrm{d}\zeta)
=\tilde{v}(t+\varepsilon,y)
=\frac{v(t,y)}{v(-\varepsilon,y_0)}
=\frac{\int_{\mathbb{S}_{\mathcal{L}}(D)} e^{-\zeta t}\,\Psi(\zeta,y)\,\nu(\dif\zeta)}{v(-\varepsilon,y_0)}.
\end{equation}
Plugging in first $y=y_0$, then $y\in D\backslash\{y_0\}$, and relying on the uniqueness of the Laplace transform (see \cite[Section 30]{billingsley2012probability}) we read off $\tilde{\nu}(\mathrm{d}\zeta)=\frac{e^{\varepsilon\zeta}}{v(-\varepsilon,y_0)}\,\nu(\mathrm{d}\zeta)$ and $\tilde{\Psi}=\Psi$ from \eqref{eq:Thm.Suff.Cond.Proof.NewIntegral.v}. Hence, for $(t,y)\in(-\varepsilon,0]\times D$,
\begin{equation}
\begin{split}
u(t,y)=v(t,y)=v(-\varepsilon,y_0)\,\tilde{v}(t+\varepsilon,y)
=v(-\varepsilon,y_0)\,\int_{\mathbb{S}_{\mathcal{L}}(D)} e^{-(t+\varepsilon)\zeta}\,\tilde{\Psi}(\zeta,y)\,\tilde{\nu}(\mathrm{d}\zeta) \\
=\int_{\mathbb{S}_{\mathcal{L}}(D)} e^{-\zeta t}\,\Psi(\zeta,y)\,\nu(\dif\zeta),
\end{split}
\end{equation}
as desired. In the special case of $y=y_0$, we obtain \eqref{eq:Lapl_conn}.
\end{proof}

\subsection{Preliminaries on positive eigenfunctions}\label{subsec:PrelimPositiveEigen}
%%%%%%%%%%%%%%%%%%%%%%%%%%%%%%%%

%If we did know the dimensions of each eigenspace corresponding to each eigenfunction though, and all of them happened to be finite, we would be able to make this choice explicitly. We use the exact same trick as when constructing the Borel measure $\nu$. Knowing the solution to the well-posed equation $\partial_t u + \mathcal{L}u =0$, as well as all of the solutions to the corresponding elliptic equations, we would choose $m-1$ points $\{y_i: 1\leq i\leq m-1,\,y_i \in D\}$. By equation \eqref{eq:Thm.Suff.Cond.Proof.NewIntegral.u} taking $m-1$ additional inverse Laplace transforms of functions $t\to u(t,y_i)$ would completely specify the only possible way of choosing our eigenfunctions. This breaks down our construction into three separate steps: find the measure, find and select the eigenfunctions, see if the integrability conditions check out. Unfortunately, in general this strategy is sometimes impossible to implement as there are cases when $\mathbb{S}_\mathcal{L}(D)$ is empty (for operator $\mathcal{L}$ in Theorem \ref{Thm.FPP.Construct} it is never empty) or $m$ is infinite.

As a preparation for the proofs of Propositions \ref{Prop.Dim1.Critical}, \ref{Prop.Dim2.One}, \ref{Prop.Dim1.Multi} and \ref{Prop.DimInfty.Multi},  we recall some facts about the sets $\mathbb{S}_{\mathcal{L}}(D)$ and $C_{\mathcal{L}-\zeta}(D)$, $\zeta\in\mathbb{S}_{\mathcal{L}}(D)$ from positive harmonic function theory. Throughout the subsection we let $\mathcal{L}$ satisfy Assumption \ref{Assump.Holder}.

\begin{definition}[Green's measure] Consider the solution $Z$ of the generalized martingale problem on $D$ associated with $\mathcal{L}_0=\mathcal{L}-P(y)$ (see \cite[Section 1.13]{pinsky1995positive}). If
\begin{equation}\label{eq:Green.def}
D'\mapsto\E\bigg[\int_0^\infty e^{\int_0^t P(Z_s)\,\mathrm{d}s}\,\mathbf{1}_{D'}(Z_t)\,\mathrm{d}t\,\bigg|\,Z_0=y\bigg]<\infty
\end{equation}
for all bounded subdomains $D'\subset D$ with $\overbar{D'}\subset D$ and $y\in D$, then the positive Borel measure defined by \eqref{eq:Green.def} is called a Green's measure for $\mathcal{L}$ on $D$. The density $G(y,z)$ of a Green's measure $G(y,\cdot)$, if it exists, is referred to as a Green's function.
\end{definition}

By \cite[Chapter 4, Theorem 3.1 and Exercise 4.16]{pinsky1995positive} for the operators $\mathcal{L}-\zeta$, $\zeta\in\rr$, we have the next proposition.

\begin{proposition}\label{Prop:GreensFunc.SolExist}
If $\zeta\in\rr$ is such that a Green's function exists for $\mathcal{L}-\zeta$, then $C_{\mathcal{L}-\zeta}(D)\neq\emptyset$.
\end{proposition}

We proceed to the corresponding classification of the operators $\mathcal{L}-\zeta$, $\zeta\in\rr$.

\begin{definition}
An operator $\mathcal{L}-\zeta$ on $D$ is described as
\begin{enumerate}[(i)]
\item subcritical if it possesses a Green's function,
\item critical if it is not subcritical, but $C_{\mathcal{L}-\zeta}(D)\neq\emptyset$,
\item and supercritical if it is neither critical nor subcritical.
\end{enumerate}
\end{definition}

Thus, we are interested in the values of $\zeta$ for which $\mathcal{L}-\zeta$ is subcritical or critical, that is, $\zeta\in \mathbb{S}_{\mathcal{L}}(D)$. As it turns out, $\mathbb{S}_{\mathcal{L}}(D)$ is a half-line under Assumption \ref{Assump.Holder}.

\begin{proposition}[\cite{pinsky1995positive}, Chapter 4, Theorem 3.2 and Exercise 4.16]\label{Prop:Eigenvalue.Spectrum.Structure}
There exists a critical eigenvalue $\zeta_c=\zeta_c(D)\in\mathbb{R}$ such that $\mathcal{L}-\zeta$ is subcritical for $\zeta>\zeta_c$, supercritical for $\zeta<\zeta_c$, and either critical or subcritical for $\zeta = \zeta_c$.
\end{proposition}

When the potential $P$ is non-positive, more information about the classification of the operator $\mathcal{L}$ is available.

\begin{proposition}[\cite{pinsky1995positive}, Chapter 4, Theorem 3.3]\label{Prop:Operator.Subcritical.BoundedPotential}
For an operator $\mathcal{L}$ with $P \leq 0$ one of the following holds:
\begin{enumerate}[(i)]
\item $P\leq 0,\, P \not \equiv 0$, and $\mathcal{L}$ is subcritical,
\item $P \equiv 0,$ the solution of the generalized martingale problem on $D$ associated with $\mathcal{L}$ is transient, and $\mathcal{L}$ is subcritical,
\item $P \equiv 0,$ the solution of the generalized martingale problem on $D$ associated with $\mathcal{L}$ is recurrent, and $\mathcal{L}$ is critical.
\end{enumerate}
\end{proposition}

\begin{remark}
When $\gamma>1$, the potential term in \eqref{L_coeff} is non-positive. This, put together with Proposition \ref{Prop:Operator.Subcritical.BoundedPotential}, yields $0 \in \mathbb{S}_{\mathcal{L}}$. Thus, $[0,\infty)\subset \mathbb{S}_{\mathcal{L}}$ by Proposition \ref{Prop:Eigenvalue.Spectrum.Structure}.
%Therefore, as specified in Theorem \ref{Thm.Suff.Cond}, we are indeed taking a Laplace transform over an open interval.
\end{remark}
%Since the potential term in our model is always non-positive, this, put together with Proposition \ref{Prop:GreensFunc.SolExist}, yields $0 \in \mathbb{S}_{\mathcal{L}}(D)$. Thus, our eigenvalue spectrum is not empty and by part (ii) of Proposition \ref{Prop:Eigenvalue.Spectrum.Structure} is a half line: $[\zeta_c, \infty) \supseteq \mathbb{S}_{\mathcal{L}}(D) \supseteq (\zeta_c, \infty)$. Therefore, as specified in Theorem \ref{Thm.Suff.Cond}, we are indeed taking a Laplace transform over an open interval.

%%%%%%%%%%%%%%%%%%%%%%%%%%%%%%%%
\subsection{Proofs of Propositions \ref{Prop.Dim1.Critical}, \ref{Prop.Dim2.One}, \ref{Prop.Dim1.Multi} and \ref{Prop.DimInfty.Multi}}\label{subsec:ProofPropositionsEigenDim}
%%%%%%%%%%%%%%%%%%%%%%%%%%%%%%%%

At this point, we can read off Propositions \ref{Prop.Dim1.Critical}, \ref{Prop.Dim2.One}, and \ref{Prop.Dim1.Multi} from appropriate results in \cite{murata1986structure} and \cite{pinsky1995positive}.

\begin{proof}[Proof of Proposition \ref{Prop.Dim1.Critical}]
By Propositions \ref{Prop:GreensFunc.SolExist} and \ref{Prop:Eigenvalue.Spectrum.Structure},
\begin{equation}
\inf\big\{\zeta\in\rr:\,
\zeta \in \mathbb{S}_{\mathcal{L}}(D)\big\}=\zeta_c(D)\in \mathbb{S}_{\mathcal{L}}(D).
\end{equation}
If $P$ is constant and the solution of the generalized martingale problem on $D$ for $\mathcal{L}-P$ is recurrent, then $\mathcal{L}-P$ is critical by Proposition \ref{Prop:Operator.Subcritical.BoundedPotential}, and hence, $\zeta_c(D)=-P$. In this case, \cite[Chapter 4, Theorem 3.4]{pinsky1995positive} yields
$|C_{\mathcal{L}-\zeta_c(D)}(D)|=1$.
\end{proof}

\begin{proof}[Proof of Proposition \ref{Prop.Dim2.One}]
It suffices to put together Proposition \ref{Prop:Eigenvalue.Spectrum.Structure} with \cite[Chapter 4, Remark 2 on p.~149, Theorem 3.4, and Exercise 4.16]{pinsky1995positive}.
\end{proof}

%\begin{proposition}[Theorem 4.3.4]\label{Prop:EllipitcSol.1Dim.CriticalEigenvalue}
%If $\mathcal{L}$ is critical on $D$, then $C_{\mathcal{L}}(D)$ is one-dimensional.
%\end{proposition}

\begin{proof}[Proof of Proposition \ref{Prop.Dim1.Multi}]
Note that, for any $\zeta\ge\zeta_c(\rr^k)$ and $f \in C_{\mathcal{L}-\zeta}$, one has $\phi f\in C_{\tilde{\mathcal{L}}-\zeta}$. Therefore, it is enough to prove $|C_{\tilde{\mathcal{L}} - \zeta}| = 1$, $\zeta>\zeta_c(\rr^k)$, which is readily obtained by combining Proposition \ref{Prop:Eigenvalue.Spectrum.Structure} with \cite[Theorem 5.3]{murata1986structure}.
\end{proof}

%\begin{proposition}\label{Prop:EllipticSol.Dim1.MutliD.Fuchs}
% Let $\mathcal{L}$ satisfy Assumptions \ref{Assump.Holder} and \ref{assump.Fuchs} on
% $\mathbb{R}^k$, where $k\geq 2$. Then if $\mathcal{L}$ is subcritical $\dim
% C_{\mathcal{L}}(\mathbb{R}^k) = 1.$
%\end{proposition}

%\begin{proposition}\label{Prop:EllipticSol.Dim2.OneDim}
%  If the operator $\mathcal{L}$ is subcritical on $D=(\alpha,\,\beta)$, where $\alpha < \beta \in [-\infty, \infty]$, then $\dim C_{\mathcal{L}}(D) = 2$.
%\end{proposition}

In the context of Proposition \ref{Prop.DimInfty.Multi}, the structure of the sets $C_{\mathcal{L}-\zeta}(D)$, $\zeta>\zeta_c(D)$ has been described in \cite[Theorems 6.1 and 6.3]{Ancona}, which we briefly recall for the convenience of the reader.

\begin{definition}[Minimal eigenfunction]
A function $\psi\in C_{\mathcal{L}-\zeta}(D)$ is referred to as minimal if $\tilde{\psi}\le\psi$ implies $\tilde{\psi}=\psi$ for all $\tilde{\psi}\in C_{\mathcal{L}-\zeta}(D)$.
\end{definition}

\begin{proposition}[\cite{Ancona}, Theorems 6.1 and 6.3] \label{Ancona}
In the setting of Proposition \ref{Prop.DimInfty.Multi}, every minimal element $\psi\in C_{\mathcal{L}-\zeta}(D)$ has the property $\lim_{z\to y} \psi(z)>0$ for exactly one point $y\in\partial D$ and is uniquely determined by $y$. In addition, for every $\psi\in C_{\mathcal{L}-\zeta}(D)$, there exists a unique Borel probability measure $\xi$ on $\partial D$ such that
\begin{equation}\label{eq:Martin_repr}
\psi(\cdot)=\int_{\partial D} \psi_y(\cdot)\,\xi(\mathrm{d}y),
\end{equation}
where $\psi_y$ is the minimal function associated with $y$.
\end{proposition}

Proposition \ref{Prop.DimInfty.Multi} is a direct consequence of Proposition \ref{Ancona}.

\begin{proof}[Proof of Proposition \ref{Prop.DimInfty.Multi}]
The uniqueness of the Borel probability measure $\xi$ in the representation \eqref{eq:Martin_repr} shows that the extreme points of $C_{\mathcal{L}-\zeta}(D)$ are precisely the minimal functions $\psi_y$, $y\in\partial D$. Clearly, $|\{\psi_y:\,y\in\partial D\}|=|\partial D|=\infty$.
\end{proof}

%%%%%%%%%%%%%%%%%%%%%%%%%%%%
\section{Merton problem in stochastic factor models}\label{sec:BackwardCase}
%%%%%%%%%%%%%%%%%%%%%%%%%%%%

In this section, we consider the framework of the Merton problem, in which an investor aims to maximize her expected terminal utility from the wealth acquired through investment:
\begin{equation}
\sup_{\pi\in\mathcal{A}}\,\E[\upsilon_T(X^\pi_T,Y_T)].
\end{equation}
Thereby, the time horizon $T$ and the utility function $\upsilon_T$ are chosen once and for all at time zero.  It is well-known (see e.g. \cite[Section IV.3]{FleSon}) that the dynamic programming equation for the Merton problem within the Markovian diffusion model \eqref{eq:Dynamic.Stock}, \eqref{eq:Dynamic.Factor} takes the shape of the HJB equation
\begin{equation}\label{eq:Merton.DynProg.PDE.HJB}
 \partial_t V + \mathcal{L}_y V - \frac{1}{2}\frac{|\lambda\,\partial_x
  V + \rho\kappa\,\partial_x\nabla_y V|^2}{\partial_{xx}V} = 0.
\end{equation}
In contrast to the preceding discussion, here the HJB equation is equipped with a terminal condition $V(T,\cdot,\cdot)=\upsilon_T$ and, hence, posed in the backward (``right'') time direction. It turns out that, under Assumption \ref{assump.Corr}, we can reduce the backward problem to a linear second order parabolic PDE posed in the ``right'' time direction, provided that the terminal utility function is of separable power form: $\upsilon_T(x,y)=\gamma^\gamma\frac{x^{1-\gamma}}{1-\gamma} g_T(y)$, and that appropriate technical assumptions hold.

\begin{theorem}\label{Prop:Backward.ValueFunction}
Let $\gamma\in(0,1)$. Suppose the market model \eqref{eq:Dynamic.Stock}, \eqref{eq:Dynamic.Factor}, the correlation matrix $\rho$, and the linear elliptic operator of the second order $\mathcal{L}$ with the coefficients
\begin{equation}\label{L_coeff'}
a(\cdot) = \kappa(\cdot)^T\kappa(\cdot),\quad
b(\cdot) = \alpha(\cdot) + \Gamma\kappa(\cdot)^T\rho^T\lambda(\cdot),\quad
P(\cdot) = \frac{\Gamma}{2q}\lambda(\cdot)^T\lambda(\cdot)
\end{equation}
satisfy the Assumptions \ref{assump.Var}, \ref{assump.Corr}, and \ref{Assump.Holder}, respectively, where
$\Gamma =\frac{1-\gamma}{\gamma}$ and $q = \frac{1}{1+\Gamma p}$. Suppose further that the Sharpe ratio $\lambda(\cdot)$ is bounded, the weak solution $Z$ of the SDE associated with $\mathcal{L}_0=\mathcal{L}-P(y)$ remains in $D$, and the terminal utility function is of separable power form $\upsilon_T(x,y)=\gamma^\gamma\frac{x^{1-\gamma}}{1-\gamma} h(y)^q$, with an $h\in C^{2,\eta}(D)$ bounded above and below by positive constants and such that
\begin{equation}
(t,y)\mapsto\nabla_y\,\E\big[e^{\int_{0}^{t} P(Z_s)\,\mathrm{d}s}\,h(Z_t)\,\big|\,Z_0=y\big]
\end{equation}
is bounded on $[0,T]\times D$. Then, the value function for the corresponding Merton problem $V(t, x, y) = \sup_{\pi\in\mathcal{A}}\,\E[\upsilon_T(X^\pi_T,Y_T)\,|\, X^\pi_T =x,\,Y_T =y]$ can be written as
\begin{equation}\label{V_from_u}
V(t,x,y)=\gamma^{\gamma} \frac{x^{1-\gamma}}{1-\gamma}u(t,y)^q.
\end{equation}
Hereby, $u$ is a classical solution of the linear PDE problem
\begin{equation}
\partial_t u + \mathcal{L} u=0\;\;\text{on}\;\;[0,T]\times D,\quad u(T,\cdot)=g_T.
\end{equation}
Moreover, every portfolio allocation $\pi^*$ fulfilling
\begin{equation}\label{eq:Prop:Backward.ValueFunction.OptPortfolio}
\sigma(Y_t) \pi_t^{*} =
\frac{1}{\gamma}\bigg(\lambda(Y_t)+q\rho\kappa(Y_t)\frac{\nabla_y u(t,Y_t)}{u(t,Y_t)}\bigg)
\end{equation}
is optimal.
\end{theorem}

\begin{proof}
By the classical verification paradigm (see e.g. \cite[Chapter IV, proof of Theorem 3.1]{FleSon}), it is enough to show that for every portfolio allocation $\pi\in\mathcal{A}$ the process $V(t,X^{\pi}_t,Y_t)$, $t\in[0,T]$ is a supermartingale, and that for every solution $\pi^*$ of \eqref{eq:Prop:Backward.ValueFunction.OptPortfolio} the process $V(t,X^{\pi^*}_t,Y_t)$, $t\in[0,T]$ is a martingale.

\medskip

We follow the proof of Proposition \ref{lem:FPP.Exists.Parab.PDE} in the reverse direction and find that $g(t,y):=u(t,y)^q$ is a classical solution of the problem \eqref{eq:MultiDPDEg}, whereas the function $V$ defined by \eqref{V_from_u} is a classical solution of the HJB equation \eqref{eq:Merton.DynProg.PDE.HJB} with $V(T,\cdot,\cdot)=\upsilon_T$. For any $\pi\in\mathcal{A}$, we may now apply It\^o's formula to $V(t,X^{\pi}_t,Y_t)$ and replace $\partial_t V+\mathcal{L}_y V$ by $\frac{1}{2}\frac{|\lambda\,\partial_x
V + \rho\kappa\,\partial_x\nabla_y V|^2}{\partial_{xx}V}$ to see that the drift coefficient of $V(t,X^{\pi}_t,Y_t)$ is the negative of the expression in \eqref{eq:drift.coeff} and, in particular, non-positive. Hence, the local martingale part of $V(t,X^{\pi}_t,Y_t)$ is bounded below by $-V(0,x,y)$ and, consequently, a supermartingale. Thus, $V(t,X^{\pi}_t,Y_t)$ is a supermartingale as well.

\medskip

Next, we deduce from the proof of Theorem \ref{Thm.Suff.Cond} that $u(t,y)$ admits the stochastic representation
\begin{equation}
u(t,y)=\E\Big[e^{\int_{0}^{T-t} P(Z_s)\,\mathrm{d}s}\,h(Z_{T-t})\,\Big|\,Z_0=y\Big]
\end{equation}
(recall that $Z$ remains in $D$ by assumption). In addition, our further assumptions imply that $\nabla_y u$ is bounded on $[0,T]\times D$, and that $u$ is bounded above and below by positive constants on $[0,T]\times D$. Together with the boundedness of the Sharpe ratio $\lambda(\cdot)$ and of $\kappa(\cdot)$ (see Assumption \ref{Assump.Holder}(iii)) this yields the boundedness of $\sigma(Y_t) \pi_t^{*}$ via \eqref{eq:Prop:Backward.ValueFunction.OptPortfolio}. Finally, the drift coefficient of $V(t,X^{\pi^*}_t,Y_t)$ vanishes and the quadratic variation of its local martingale part computes to
\begin{equation}
\begin{split}
\int_0^t
\gamma^{2\gamma}\,(X^{\pi^*}_s)^{2-2\gamma}\,|\sigma(Y_s)\pi^*_s|^2
+\frac{\gamma^{2\gamma}q^2}{(1-\gamma)^2}\,(X^{\pi^*}_s)^{2-2\gamma}\,u(s,Y_s)^{2q-2}\,|\kappa(Y_s)\nabla_y u(s,Y_s)|^2 \\
+\frac{2\gamma^{2\gamma}q}{1-\gamma}\,(X^{\pi^*}_s)^{2-2\gamma}\,u(s,Y_s)^{q-1}\,(\sigma(Y_s)\pi^*_s)^T\rho\kappa(Y_s)\nabla_y u(s,Y_s)\,\mathrm{d}s.
\end{split}
\end{equation}
The expectation of the latter integral is finite for all $t\in[0,T]$, since $\sigma(Y_s)\pi^*_s$ and $u(s,Y_s)^{q-1}\kappa(Y_s)\nabla_y u(s,Y_s)$ are bounded, while $\sup_{t\in[0,T]} \E[(X^{\pi^*}_t)^{2-2\gamma}]<\infty$ thanks to the boundedness of $\sigma(Y_s)\pi^*_s$ and $\lambda(Y_s)$ in
\begin{equation}
X^{\pi^*}_t=x\,\exp\bigg(\int_0^t (\sigma(Y_s)\pi^*_s)^T\lambda(Y_s)\,\mathrm{d}s+\int_0^t (\sigma(Y_s)\pi^*_s)^T\,\mathrm{d}W_s-\frac{1}{2}\int_0^t |\sigma(Y_s)\pi^*_s|^2\,\mathrm{d}s\bigg).
\end{equation}
We conclude that $V(t,X^{\pi^*}_t,Y_t)$ is a true martingale.
\end{proof}

%%%%%%%%%%%%%%%%%%%%%%%%%%%%%
\section{Discussion of EVE assumption}\label{sec: corr matr}
%%%%%%%%%%%%%%%%%%%%%%%%%%%%%

This last section is devoted to a thorough investigation of Assumption \ref{assump.Corr} that plays a key role in the proof of Theorem \ref{Thm.FPP.Construct}. It is instructive to start with the two extreme cases corresponding to taking $p=1$ and $p=0$ therein, respectively. Suppose first that $A=0$ in \eqref{eq:BM}, in other words, the components of the Brownian motion $B$ driving the factors are given by linear combinations of the components of the Brownian motion $W$ driving the stock prices. We can then reparametrize the model such that $B=W$, $\rho=I_{d_W}$, and $\rho^T\rho=I_{d_W}$. Consequently, Assumption \ref{assump.Corr} holds with $p=1$. The resulting market is complete, and we find ourselves in the framework of \cite[Section 2.3]{nadtochiy2015optimal}. It is therefore not surprising that the HJB equation \eqref{eq:Thm.FPP.Sufficient.PDE.HJB} can be reduced to a linear PDE, even though the linearization in Proposition \ref{lem:FPP.Exists.Parab.PDE} differs from the one in \cite[Section 2.3]{nadtochiy2015optimal}. On the other hand, when $\rho=0$ in \eqref{eq:BM}, the Brownian motions $B$ and $W$ become independent, leading to an incomplete market. Nonetheless, Assumption \ref{assump.Corr} is still satisfied with $p=0$. Thus, the linearization in Proposition \ref{lem:FPP.Exists.Parab.PDE} goes far beyond the complete market setup.

\medskip

More generally, Assumption \ref{assump.Corr} can be put to use as follows. In practice, the correlation matrix $\rho$ can have hundreds or thousands of entries, and hence, might be difficult to estimate accurately in its entirety. However, one can attempt to obtain a less noisy estimate by projecting an estimate for $\rho$ onto the submanifold of $d_W\times d_B$ matrices fulfilling Assumption \ref{assump.Corr}. Restricting the attention to the non-trivial case $d_W\ge d_B$ (see Remark \ref{rmk:EVE_disc}), with the exception of the zero matrix, the latter matrices can be written uniquely as $rQ$, where $r\in(0,1]$ and $Q$ has orthonormal columns, thereby forming a $\big(1+\frac{d_W(d_W-1)}{2}-\frac{(d_W-d_B)(d_W-d_B-1)}{2}\big)$-dimensional submanifold of $\rr^{d_W\times d_B}$. As it turns out, the most tractable projection onto this submanifold is that with respect to the Frobenius norm (also known as the Hilbert-Schmidt norm) on $\rr^{d_W\times d_B}$.

%\textbf{L: The stock-factor correlation matrix can be thought of as a biadjacency matrix of a weighted bipartite graph, where the vertices are represented by the coordinates of the Brownian motions $B$ and $W$. If one wants to study this matrix, I assume it could be a good idea to see which norms are used for such biadjacency matrices, and study their properties in general to find things that could describe our model.}

%%%%%%%%%%%%%%%%%%%%%%%%%%%
\subsection{Choice of $r$ and $Q$}
%%%%%%%%%%%%%%%%%%%%%%%%%%%

Let us equip the space $\rr^{d_W\times d_B}$ with the Frobenius norm:
\begin{equation}
|A|_F = \bigg(\sum_{i=1}^{d_W} \sum_{j=1}^{d_B} a_{ij}^2\bigg)^{1/2}
=\big(\mathrm{trace}\,A^TA\big)^{1/2}.
\end{equation}
For an estimate $\widehat{\rho}$ of $\rho$, we are able to find a constant $r$ and a matrix with orthonormal columns $Q$ that minimize the distance defined by the Frobenius norm.
\begin{proposition}
Consider the minimization problem
\begin{equation}
\min |\widehat{\rho}-rQ|_F\;\;\textrm{such that}\;\;r\in[0,1],\;\;Q^TQ=I_{d_B}.
\end{equation}
Then, $r^* = \frac{\mathrm{trace} (\widehat{\rho}^T\widehat{\rho})^{1/2}}{d_B}$ and $Q^*= \widehat{\rho}(\widehat{\rho}^T\widehat{\rho})^{-1/2}$ are the minimizers.
\end{proposition}
\begin{proof}
Equivalently, consider the problem
\begin{equation}
\min |\widehat{\rho} - \widetilde{Q}|^2_F\;\;\textrm{such that}\;\;\widetilde{Q}^T\widetilde{Q}=r^2 I_{d_B}
\end{equation}
for fixed $r\in[0,1]$ and minimize over $r\in[0,1]$ subsequently. Applying the method of Lagrange multipliers with a $d_B\times d_B$ Lagrange multiplier matrix $\Lambda$ we get
\begin{equation}\label{optimal_Q0}
2(\widetilde{Q}-\widehat{\rho}) = \widetilde{Q}(\Lambda+\Lambda^T) \quad \Longleftrightarrow \quad \widetilde{Q}(2I_{d_B}-\Lambda-\Lambda^T)=\widehat{\rho}.
\end{equation}
Passing to the transpose on both sides of the last equation, taking the product of the resulting equation with the original equation, and recalling the constraint we see
\begin{equation}
r^2(2I_{d_B}-\Lambda-\Lambda^T)^2=\widehat{\rho}^T\widehat{\rho}\quad \Longleftrightarrow\quad
r(2I_{d_B}-\Lambda-\Lambda^T)=(\widehat{\rho}^T\widehat{\rho})^{1/2},
\end{equation}
where $(\widehat{\rho}^T\widehat{\rho})^{1/2}$ is a $d_B\times d_B$ square root of the matrix $\widehat{\rho}^T\widehat{\rho}$. Together with \eqref{optimal_Q0} and the notation $(\widehat{\rho}^T\widehat{\rho})^{-1/2}$ for the inverse of $(\widehat{\rho}^T\widehat{\rho})^{1/2}$ this yields
\begin{equation}\label{optimal_Q}
\widetilde{Q}=r\widehat{\rho}(\widehat{\rho}^T\widehat{\rho})^{-1/2}.
\end{equation}
Plugging the formula for $\widetilde{Q}$ back into the objective function we are left with the minimization problem
\begin{equation}
\min_{r\in[0,1]} \big|\widehat{\rho} - r\widehat{\rho}(\widehat{\rho}^T\widehat{\rho})^{-1/2}\big|^2_F \quad \Longleftrightarrow \quad
\min_{r\in[0,1]} \big(\mathrm{trace} (\widehat{\rho}^T\widehat{\rho})-2r\,\mathrm{trace} (\widehat{\rho}^T\widehat{\rho})^{1/2}+r^2 d_B\big).
\end{equation}
Consequently, the optimal $r$ is $\frac{\mathrm{trace} (\widehat{\rho}^T\widehat{\rho})^{1/2}}{d_B}$, that is, the average of the singular values of $\widehat{\rho}$, whereas $\widetilde{Q}$ should be picked according to \eqref{optimal_Q}.
\end{proof}

%%%%%%%%%%%%%%%%%%%%%%%%%%%
\subsection{Choice of $p$}
%%%%%%%%%%%%%%%%%%%%%%%%%%%

%Denote the singular values of $\rho$ as $\{m_i\}_{i=1}^{d_B}$ in a descending order and denote the corresponding diagonal matrix as $M \in \mathbb{R}^{d_B\times d_B}$. For now (\textbf{WLOG?}) assume that all of the singular values are non-zero. Consider matrices $U \in \mathbb{O}^{d_W \times d_B}$ and $V \in \mathbb{O}^{d_B \times d_B}$ such that $\rho = U M V^T$. The minimization problem we are facing is:
%\begin{align*}
%  \min_{r\in[-1,1],\, Q \in \mathbb{O}^{d_W \times d_B}}\|\rho - rQ\|_2 \quad \Leftrightarrow \min_{r\in[-1,1],\, Q \in \mathbb{O}^{d_W \times d_B}}R(\rho^T\rho - 2r \rho^T Q + r^2 I_{d_B\times d_B}).
%\end{align*}
%Guess: $Q = UV^T,\, r = \frac{m_1 + m_{d_B}}{2}$ is the solution.
%2x2 case shows that guess is incorrect and correct solution ugly.

If one is only interested in the parameter $p$ from Assumption \ref{assump.Corr}, then it is most natural to minimize $|\widehat{\rho}^T\widehat{\rho}-pI_{d_B}|$ for a selection of a norm $|\cdot|$ on $\rr^{d_B\times d_B}$. When $|\cdot|$ is the operator norm (also known as the spectral radius or the Ky Fan $1$-norm),
\begin{equation}
|\widehat{\rho}^T\widehat{\rho}-pI_{d_B}|=\max_{1\le i\le d_B} |\theta_i-p|,
\end{equation}
where $\theta_1\le\theta_2\le\cdots\le\theta_{d_B}$ are the ordered eigenvalues of $\widehat{\rho}^T\widehat{\rho}$ (or, equivalently, the ordered squared singular values of $\widehat{\rho}$). In this case, $|\widehat{\rho}^T\widehat{\rho}-pI_{d_B}|$ is minimized by $p=\frac{\theta_1+\theta_{d_B}}{2}$. When $|\cdot|$ is the Frobenius norm,
\begin{equation}
|\widehat{\rho}^T\widehat{\rho}-pI_{d_B}|=\bigg(\sum_{i=1}^{d_B} |\theta_i-p|^2\bigg)^{1/2},
\end{equation}
which is smallest for $p=\frac{\theta_1+\theta_2+\cdots+\theta_{d_B}}{d_B}$. When $|\cdot|$ is the trace norm (also known as the nuclear norm or the Ky Fan $d_B$-norm),
\begin{equation}
|\widehat{\rho}^T\widehat{\rho}-pI_{d_B}|=\sum_{i=1}^{d_B} |\theta_i-p|.
\end{equation}
The minimizer $p$ for the latter is the median of $\{\theta_1,\theta_2,\ldots,\theta_{d_B}\}$.

%%%%%%%%%%%%%%%%%%%%%%%%%%%
\subsection{Example: affine factor models}
%%%%%%%%%%%%%%%%%%%%%%%%%%%

We conclude by illustrating the use of the EVE assumption in the framework of affine market models with non-negative factors. In that situation, both the forward investment problem and the Merton problem can be reduced to the solution of a system of Riccati ordinary differential equations (ODEs). Consider the affine specialization of the factor model \eqref{eq:Dynamic.Stock}-\eqref{eq:BM}:
\begin{eqnarray}
&& \frac{\mathrm{d}S^i_t}{S_t^i}=\mu_i(Y_t)\,\mathrm{d}t+\sum_{j=1}^{d_W} \sigma_{ji}(Y_t)\,\mathrm{d}W^j_t,\quad i=1,\,2,\,\ldots,\,n, \label{eq:Dynamic.Stock.Affine} \\
&& \dif Y_t = (M^T Y_t + w) \dif t + \kappa(Y_t)^T\dif B_t,
\label{eq:Dynamic.Factor.Affine} \\
&& B_t = \rho^T W_t + A^T W^{\perp}_t,
\end{eqnarray}
where $M$ has non-negative off-diagonal elements, $w \in [0, \infty)^k$, and  $\mu(\cdot)$, $\sigma(\cdot)$, $\kappa(\cdot)$, $\rho$ are such that
\begin{eqnarray}
% \nonumber % Remove numbering (before each equation)
&& \lambda(y)^T\lambda(y) = \mu(y)^T \big(\sigma(y)\big)^{-1} \big(\sigma(y)^T\big)^{-1} \mu(y) = \Lambda y + \lambda_0, \label{eq:Affine.Assump.Sharpe} \\
&& \kappa(y)^T\kappa(y) = \mathrm{diag}(L_1y_1,L_2y_2,\ldots,L_k y_k)\;\;\text{with}\;\;L_1,L_2,\ldots,L_k > 0, \label{eq:Affine.Assump.Vol} \\
&& \Gamma \kappa(y)^T \rho^T \lambda(y) = N^T y + c. \label{eq:Affine.Assump.DriftExtra}
\end{eqnarray}
\begin{remark}
The condition \eqref{eq:Affine.Assump.Vol} is necessary for the process $Y$ of \eqref{eq:Dynamic.Factor.Affine} to be $[0,\infty)^k$-valued and affine (see \cite[Theorem 3.2]{filipovic2009affine}). Conversely, the SDE \eqref{eq:Dynamic.Factor.Affine} with volatility and drift coefficients satisfying
\eqref{eq:Affine.Assump.Vol} has a unique weak solution, which is affine and takes values in $[0,\infty)^k$ (see \cite[Theorem 8.1]{filipovic2009affine}).
\end{remark}
Suppose now that the initial utility function for the forward investment problem or the terminal utility function for the Merton problem is of separable power form with $h(y)=\exp(H^Ty+h_0)$. Under the EVE assumption, the HJB equation \eqref{eq:Thm.FPP.Sufficient.PDE.HJB} arising in the two problems can be transformed into the linear second-order parabolic PDE of \eqref{problem.for.Phi} (see the proof of Proposition \ref{lem:FPP.Exists.Parab.PDE}), which in the setting of \eqref{eq:Dynamic.Stock.Affine}-\eqref{eq:Affine.Assump.DriftExtra} amounts to
\begin{equation}\label{eq:Affine.Parab.PDE}
\partial_t u + \frac{1}{2} \sum_{i=1}^{k} L_i y_i\partial_{y_i,y_i}u + y^T (M+N) \nabla_y u + (w + c)^T \nabla_y u + \frac{\Gamma}{2q}(\Lambda^T y+ \lambda_0) u = 0.
\end{equation}
Inserting the exponential-affine ansatz $u(t,y)=\exp(\Phi(t)^Ty+\Theta(t))$ we obtain
\begin{equation}
y^T\dot{\Phi}(t) + \dot{\Theta}(t) + \frac{1}{2}\sum_{i=1}^{k}L_i y_i \Phi_{i}^2 + y^T (M+N) \Phi(t) + (w+c)^T\Phi(t) + \frac{\Gamma}{2q}(\Lambda^T y+\lambda_0)=0.
\end{equation}
Equating the linear and the constant terms in $y$ to $0$ leads to the following system of Riccati ODEs:
\begin{eqnarray}
&&\dot{\Phi}_i(t) + \frac{1}{2} L_i\Phi_i(t)^2 + \sum_{j=1}^{k}(M + N)_{ij} \Phi_j(t) + \frac{\Gamma}{2q}\Lambda_i = 0, \quad i=1,\,2,\,\ldots,\,k, \qquad\qquad \label{sys:Affine.RiccatiODE.Phi} \\
&&\dot{\Theta}(t) + (w + c)^T\Phi(t)  + \frac{\Gamma}{2q}\lambda_0 = 0. \label{eq:Affine.RiccatiODE.Theta}
\end{eqnarray}
We note that $\Theta$ is completely determined by the solution $\Phi$ of the system \eqref{sys:Affine.RiccatiODE.Phi}. The latter can be solved numerically in general and, for special kinds of $M$ and $N$, even explicitly. For example, when $M$ and $N$ are diagonal the system \eqref{sys:Affine.RiccatiODE.Phi} splits into $k$ one-dimensional Riccati ODEs:
\begin{equation}\label{one_dim_ric}
\dot{\Phi}_i(t) + \frac{1}{2} L_i\Phi_i(t)^2 + (M_{ii} + N_{ii}) \Phi_i(t) + \frac{\Gamma}{2q}\Lambda_i = 0, \quad i=1,\,2,\,\ldots,\,k.
\end{equation}
These ODEs can be solved by a separation of variables and subsequent integration. For instance, when $\gamma>1$ and $\Lambda_i>0$ for all $i$, the discriminants
$D_i:=(M_{ii}+N_{ii})^2-L_i\frac{\Gamma }{q}\Lambda_i$ associated with the quadratic equations $\frac{1}{2}L_i z^2 + (M_{ii} + N_{ii}) z + \frac{\Gamma}{2q}\Lambda_i=0$ are positive, resulting in the roots
\begin{equation}
z_{+,i} = \frac{-M_{ii} - N_{ii} + \sqrt{D_i}}{L_i},\quad z_{-,i} = \frac{-M_{ii} - N_{ii}- \sqrt{D_i}}{L_i}.
\end{equation}
The general solution of \eqref{one_dim_ric} then becomes
\begin{equation}
% \nonumber % Remove numbering (before each equation)
\Phi_i(t) = \frac{z_{+,i}-\chi_i\,z_{-,i}\,e^{-\sqrt{D_i}t}}{1-\chi_i\, e^{-\sqrt{D_i}t}},\quad i=1,\,2,\,\ldots,\,k, \label{eq:Affine.His0.Sol.Phi}
\end{equation}
and one can find the constants $\chi_i$ by setting $\Phi(\cdot)$ to $H$ at time $0$ (for the forward investment problem) or at the terminal time (for the Merton problem).

\medskip

\printbibliography
\end{document}